\documentclass[aps,prd,twocolumn,showpacs,preprintnumbers,amsmath,superscriptaddress,amssymb,floats,nofootinbib]{revtex4-1}
\usepackage{graphicx,color,float}
\usepackage{amsthm}

\begin{document}
\newcommand {\T} {{\cal T}}
\newcommand {\DC} {{\delta \chi^2}}
\newcommand {\A} {{A}}
\newcommand {\al} {c} 
\newcommand {\bl} {b}
\newcommand {\erf} {\text{erf}}
\newcommand {\ttheta} {\theta} 
\newcommand {\x}  {{\theta_{\text{true}}}}
\newcommand {\y}  {{\theta_{\min}}}
\newcommand {\Tm} {{\Theta_{\min}}}
\newcommand {\Tt} {{\Theta_{\text{true}}}}
\newcommand {\Eta} {{\cal H}}

\newcommand {\hmu} {\hat{\mu}}
\newcommand {\hnu} {\hat{\nu}}
\newcommand {\hpi} {\hat{\pi}}
\newcommand {\htau} {\hat{\tau}}
\newcommand {\hth} {\hat{\theta}} 
\newcommand {\dg} {\text{diag}}
\newcommand {\p} {p}
\newcommand {\s} {s}
\newcommand {\ba} {{\bf a}}
\newcommand {\bb} {{\bf b}}
\newcommand {\bc} {{\bf c}}
\newcommand {\0} {{\bf 0}}
\newcommand {\f} {{\bf f}}
\newcommand {\e} {{\bf e}}
\newcommand {\B} {{B^*}}
\newcommand {\itau} {M^{-1}}
\newcommand {\mtau} {{M}}

\newcommand {\rf} {\text{ref}}
\newcommand {\tr} {\text{true}}
\newcommand {\NH} {{H_0}}
\newcommand {\IH} {{H_1}}

\newcommand {\Prob} {\text{Prob}}
\newcommand {\E} {\text{E}}
\newcommand {\Var} {\text{Var}}
\newcommand{\ind}{\stackrel{\text{indep}}{\sim}}
\newcommand{\cas}{\buildrel \text{a.s.} \over \longrightarrow}
\newcommand{\casleft}{\buildrel \text{a.s.} \over \longleftarrow}
\newcommand{\cd}{\buildrel d \over \rightarrow}
\newcommand{\cp}{\buildrel P \over \longrightarrow}
\newcommand{\apsim}{\buildrel \text{approx.} \over \sim}

\newtheorem{lemma}{Lemma}

\newcommand {\Tbarn} {\overline{\Delta T_{H_0}}}
\newcommand {\Tbara} {\overline{\Delta T_{H_1}}}
\newcommand {\n} {{n}}
\newcommand {\Ni} {{N_i}}
\newcommand {\nutn} {\eta} 
\newcommand {\nuta} {\zeta} 
\newcommand {\nutnT} {\eta^0} 
\newcommand {\nutaT} {\zeta^0} 
\newcommand {\muT} {\mu^0} 
\newcommand {\nuT} {\nu^0} 
\newcommand {\piT} {{\pi^{0}}} 
\newcommand {\tauT} {\tau^0} 
\newcommand {\tmu} {\tilde{\mu}}
\newcommand {\tnu} {\tilde{\nu}}

\newcommand {\hnutn} {\hat{\nutn}}
\newcommand {\hnuta} {\hat{\nuta}}

\newcommand {\paratheta} {\beta} 
\newcommand {\PT} {B}

\newcommand {\D} {\overline{D}}

\def\Journal#1#2#3#4{{#1} {\bf #2}, #3 (#4)}
\def\NCA{\rm Nuovo Cimento}
\def\NPA{{\rm Nucl. Phys.} A}
\def\NIM{\rm Nucl. Instrum. Methods}
\def\NIMA{{\rm Nucl. Instrum. Methods} A}
\def\NPB{{\rm Nucl. Phys.} B}
\def\PLB{{\rm Phys. Lett.}  B}
\def\PRL{\rm Phys. Rev. Lett.}
\def\PRD{{\rm Phys. Rev.} D}
\def\PRC{{\rm Phys. Rev.} C}
\def\ZPC{{\rm Z. Phys.} C}
\def\JPG{{\rm J. Phys.} G}


\title{The Gaussian CL$_s$ Method for Searches of New Physics}

\date{\today}
\author{X. Qian}\email[Corresponding author: ]{xqian@bnl.gov}
\affiliation{Brookhaven National Laboratory, Upton, NY}
\author{A. Tan}\email[Corresponding author: ]{aixin-tan@uiowa.edu}
\affiliation{Department of Statistics and Actuarial Science, University of Iowa, Iowa City, IA}
\author{J. J. Ling}
\affiliation{Department of Physics, University of Illinois of Urbana-Champaign, Urbana, IL}
\author{Y. Nakajima}
\affiliation{Lawrence Berkeley National Laboratory, Berkeley, CA}
\author{C. Zhang}
\affiliation{Brookhaven National Laboratory, Upton, NY}



\begin{abstract} 
We describe a method based on the CL$_s$ approach to present results in searches 
of new physics, under the condition that the relevant parameter space is continuous. 
Our method relies on a class of test statistics developed for non-nested hypotheses testing problems, denoted by $\Delta T$, which has a Gaussian approximation to its parent distribution when the sample size is large. This leads to a simple procedure of forming exclusion sets for 
the parameters of interest, which we call the Gaussian CL$_s$ method. 
Our work provides a self-contained mathematical proof for the Gaussian CL$_s$ method, that explicitly outlines the required conditions. These conditions are milder 
than that required by the Wilks' theorem to set confidence intervals (CIs).
We illustrate the Gaussian CL$_s$ method in an example of searching for a sterile neutrino, where the CL$_s$ approach was rarely used before. We also compare data analysis results produced by the Gaussian CL$_s$ method and various CI methods to showcase their differences.

 \end{abstract}

\maketitle
\thispagestyle{plain}
\section{Introduction}\label{sec:introduction}

The Standard Model of particle physics has been extremely successful since 
its establishment in the mid-1970s. In particular, the Higgs particle 
discovered at LHC in 2012~\cite{ATLAS_higgs,CMS_higgs} completed the 
list of fundamental particles predicted by the minimal Standard Model. 
On the other hand, there has been experimental evidence that point to new physics beyond the Standard Model: 
neutrino oscillations indicate non-zero neutrino mass; various gravitational effects indicate 
the existence of non-baryonic dark matter; the accelerating expansion 
of our universe indicates the existence of dark energy; the large 
observed matter-anti-matter asymmetry in the universe indicates the existence of 
additional CP violation source beyond that in the quark mixing matrix, etc. 
Searches for new physics beyond the Standard Model have been and still are 
at the frontier of high energy particle physics. 

Given experiment data, a problem of searching for new physics often turns into a parameter estimation problem, and the findings are presented in the form of constraints on some continuous parameter(s). One example is the search for sterile neutrino suggested by LSND~\cite{LSND}, MiniBooNE~\cite{miniboone}, and reactor antineutrino anomalies~\cite{anom}.~\footnote{Other examples include
  dark matter searches (the interaction cross section vs. the mass of
  the dark matter particle) and SUSY (super symmetry) particle
  searches at LHC (the interaction coupling vs. the mass scale).} In this case, data collected from an experiment consists of neutrino interaction counts in multiple energy bins, $x=(N_1,\cdots,N_n)$. Data analysis results are generally shown as constraints in the two-dimensional parameter space of $(\sin^22\theta, |\Delta m^2|)$, where $\theta$ is the mixing angle involving the sterile neutrino, and $|\Delta m^2|$ is the mass-squared difference of neutrino mass eigenstate beyond three generations. 



One way to set constraints is to form confidence intervals\footnote{Rigorously speaking, the word ``confidence set" should be used 
instead of ``confidence interval" when the dimension of the parameter space is higher 
than one. But as long as there is no ambiguity, we will refer to all confidence sets 
as confidence intervals for simplicity.} (CI), which contains parameter values that are compatible with the data. Let $\paratheta$ denote the parameter(s), such as $\paratheta=(\sin^22\theta, |\Delta m^2|)$ in the neutrino oscillation problem. A CI can be obtained by inverting a testing procedure. Specifically, the set of all $\paratheta_1$ such
that the hypothesis $H_0: \paratheta=\paratheta_1$ is not rejected at
level $1-c$, forms a level-$c$ CI. A testing procedure is often performed by thresholding a test statistic,
which is a user-chosen function that, for any given $\paratheta_1$,
defines a criterion to order all possible values of $x$. To test $H_0:
\paratheta=\paratheta_1$,
a commonly used type of test statistic takes the form
  \begin{equation}\label{eq:nestchi2}\Delta \chi^2
  (\paratheta_1;x):= \chi^2 (\paratheta_1;x) -
\chi^2_{\min}(x),\end{equation}
 where $\chi^2$ is a function that measures the compatibility between
$\paratheta_1$ and $x$. One important example of $\chi^2$ is the negative-two-log-likelihood 
function, and the corresponding $\Delta \chi^2$ is called the likelihood ratio (LR) test statistic. 
In the field of high energy physics, the unified approach to construct CIs advocated by 
Feldman and Cousins~\cite{Feldman:1997qc} is indeed based on the likelihood ratio test statistic. 

A parameter value $\paratheta_1$ is included in a level-$\al$ CI if $\Delta \chi^2 (\paratheta_1; x)$ is below a threshold $t_\al$, such that $\Prob_{\paratheta_1}(\Delta \chi^2 (\paratheta_1; X) \leq t_\al)\geq \al$. Here, the subscript $\paratheta_1$ means that $X$ is a random outcome from a model 
with true parameter value $\paratheta_1$.  In general, Monte Carlo (MC) simulation can be used~\footnote{When the model contains nuisance 
parameters, extra care are needed in performing Monte Carlo simulation. See Ref.~\cite{full_freq} for example.} 
to approximate the parent distribution of $\Delta \chi^2$. We refer to the corresponding method 
of constructing CIs as {\it the MC CI method}.  An example of the MC CI method, tailored for the 
LR test statistic, can be found in section~V.B of Feldman and Cousins~\cite{Feldman:1997qc}. 
The MC CI method is often computationally intensive. Alternatively, $t_\al$ can be approximated using a Chi-square distribution, a summary of its usage in particle physics is provided by the Particle Data Group~\cite{PDG}. This method is simple to carry out, but the approximation is only valid under relatively stringent conditions. Specifically, the Chi-square thresholds are justified by the Wilks' theorem~\cite{Wilks-1938} 
for the LR test statistic under regularity conditions C1---C3 in 
Sec.~\ref{sec:review}, and they are 
justified for the variations of $\Delta \chi^2$ listed in Sec.~\ref{sec:CLs_1} under similar conditions~\cite{wald:1943, silvey:1959}. We conveniently refer to any method that constructs approximate CIs based on Chi-square thresholds as 
a {\it Wilks' CI method}. 

In theory, forming CIs using test statistics of the form $\Delta \chi^2 (\paratheta; X)$ 
is desirable, 
because it leads to a unified approach in setting limits in the absence of new 
physics signals and in estimating parameters after the discovery of new physics~\cite{Feldman:1997qc}. 
However, in the problem of searching for sterile neutrinos, the computationally 
expensive MC CI method is usually necessary to obtain valid thresholds $t_c$ for the 
$\Delta \chi^2$ statistic, making the application difficult.

Compared to $\Delta \chi^2$, the following test statistic, $\Delta T$, has a 
parent distribution that is easy to approximate under mild conditions. 
By fixing a reference value of $\paratheta$, 
say $\paratheta_\rf$,  one can test a pair of non-nested hypotheses 
$H_0: \paratheta=\paratheta_\rf$ versus $H_1: \paratheta = \paratheta_1$ using a test 
statistic of the form
\begin{equation}\label{eq:nonnestT}
\Delta T(\paratheta_\rf,\paratheta_1;x):= \chi^2 (\paratheta_1;x) - \chi^2 (\paratheta_\rf;x) \,.
\end{equation}
An example of $\Delta T$ is the negative-two-log-likelihood ratio statistic for $H_0$ and $H_1$.
Given observed data $x$ and a fixed $\paratheta_\rf$, all $\beta_1$ values that result in $\Prob_{\paratheta_1}(\Delta T(\paratheta_\rf, \paratheta_1; X) \leq t'_\al)\geq \al$ constitute a  level-$\al$ CI. It is proven in Sec.~\ref{sec:Gaussian} that under fairly mild conditions, 
one can approximate $t'_\al$ using quantiles of a Gaussian distribution. Specifically, we show that when the data size is large, 
the distribution of $\Delta T(X)$, where $X$ represents potential data from a model that 
satisfies either one of the two hypotheses, say $H$, can be approximated by the Gaussian 
distribution with mean $\overline{\Delta T_H}$ and standard deviation 
$2\sqrt{|\overline{\Delta T_H}|}$. Here, $\overline{\Delta T_H}$ is defined to be 
$\Delta T(x_H^{\text{Asimov}})$ as in Eq.~\eqref{eq:deltaT}, where $x_H^{\text{Asimov}}$ is the 
{\it Asimov data set}~\cite{asimov} as introduced in Sec.~\ref{sec:DeltaT}. 


However, CIs constructed from $\Delta T$ can exclude $\paratheta_1$ values that are not much less compatible with the data than $\paratheta_\rf$ is, 
which we demonstrate in Sec.~\ref{sec:CLvsCI}. To avoid counter-intuitive results based on $\Delta T$, we take the CL$_s$ approach of setting {\it exclusion sets}~\cite{Read:2000ru,Junk:1999kv,Read:2002hq} as an alternative to the CI approach. We refer to the simple procedure of setting exclusion sets based on the $\Delta T$ statistic using a Gaussian approximation as {\it the Gaussian CL$_s$ method}. 

Note that an exclusion set imposes a different kind of constraint than that of (the complement of) a CI. An exclusion set aims at identifying parameter values that fit the data much worse than the reference model. Consequently, the CL$_s$ approach is more reluctant than the CI approach to exclude models where the experiment has little sensitivity. An example comparing the two can be found in Sec.~\ref{sec:CLvsCI}. 

The main contribution of this paper is to provide a mathematical proof for a Gaussian 
approximation to the distribution of $\Delta T$.  This result justifies the Gaussian CL$_s$ method, which requires a computational load similar to that of the Wilks' CI method, and the former is valid in situations where the latter is not. Results similar to ours can be found in Ref.~\citep{asimov} in the context of searching new particles, and in Ref.~\cite{Qian:2012zn,Blennow:2013oma} in the context of neutrino mass hierarchy determinations. The self-contained proof provided in this paper makes it easier 
to fully articulate the required conditions, which were missing in the previous work.  
Also, we make a more general and realistic assumption in accordance with the physics 
problem of interest than that of Ref.~\citep{asimov} and the paper by Wald~\citep{wald:1943} cited therein. For details, see assumptions [A0] and [A1] in Sec.~\ref{sec:approxT}.

Another contribution of this paper is that we compare various methods 
that take the CI approach or the CL$_s$ approach in a problem of searching for 
neutrino oscillations, where the CL$_s$ approach was rarely used before. 
Based on the comparisons, we advocate the Gaussian CL$_s$ as an attractive alternative 
method to the CI approach in the 
application of searching for new physics through precision measurements. 
First, the Gaussian CL$_s$ 
method is inexpensive to carry out and is valid in very general setups. Secondly, researchers often need to combine results from different experiments. When conditions in the Wilks' theorem are not satisfied, it is simple to combine the test statistics from different experiments and form an overall CI using the Wilks' method.   
Otherwise, expensive MC methods have to be used to form CIs for each experiment, and there is no rigorous way to combine these results together other than to rerun a more expensive MC for the combined 
data. 
In contrast, we explain in Sec.~\ref{sec:discussion} that experimental results can be easily combined using the Gaussian CL$_s$ method, and is valid under mild conditions.

This paper is organized as follows. In Sec.~\ref{sec:review}, 
we briefly review the CI approach that utilizes a class of statistics, $\Delta \chi^2$. We look at both the Wilks' CI method and the MC CI method, and discuss their advantages and limitations. 
In Sec.~\ref{sec:DeltaT}, we describe an alternative class of statistics, $\Delta T$. 
In Sec.~\ref{sec:CLs}, we describe the CL$_s$ approach based on the $\Delta T$ statistic, 
and outline a simple procedure to carry it out using the Gaussian approximation.
In Sec.~\ref{sec:example}, using an example of the search for a sterile neutrino, we  
check the validity of the approximation in the Gaussian CL$_s$ method, and compare 
different methods of forming constraints in the parameter space. Finally, we present discussions and 
summaries in Sec.~\ref{sec:discussion} and Sec.~\ref{sec:summary}, respectively.

\section{The Confidence Interval approach based on the $\Delta \chi^2$ statistic}~\label{sec:review}

In this section, we briefly review the traditional method of setting 
CIs in the context of neutrino oscillations. We consider 
a neutrino energy spectrum that consists of $n$ energy bins, and assume that 
the mean number of counts in each bin is a function of the vector of 
parameters of main interest, $\paratheta=(\sin^22\theta,|\Delta m^2|)$, 
and a vector of nuisance parameters (such as the overall normalization), 
$\eta$. Let $\Theta$, $M$, and $\Eta$ denote the parameter space of 
$\sin^22\theta$, $|\Delta m^2|$, and $\eta$, respectively. 
There are two further physical constraints: $\sin^22\theta \ge 0$ and 
$|\Delta m^2| \ge 0$. Then for the $i$-th bin, 
$\lambda_{i}(\sin^22\theta, |\Delta m^2|, \eta)$ and $N_i$ represent the 
mean and the observed counts of neutrino induced interactions, 
respectively.  When $\lambda_{i}$ is large enough, the distribution of 
$N_i$ can be well approximated by a Gaussian distribution with mean 
$\lambda_{i}$ and standard deviation $\sqrt{\lambda_{i}}$.

Given any specific guess of the value of the parameters 
$(\sin^22\theta, |\Delta m^2|, \eta)$, once the data $x=\{N_i, i=1,\ldots,n\}$ 
are observed, one can calculate the deviation of the data from 
the mean values $\lambda_i$ to measure the compatibility
 of the hypothesized parameter values to $x$. 
Commonly used deviations include negative-two-log-likelihood ratio, Pearson chi-square and Neyman chi-square. Further, when certain knowledge 
concerning the nuisance parameter $\eta$ 
(e.g. knowledge of detecting efficiency and neutrino flux) 
is available, it can be reflected in the definition of the deviation. 
For example, to modify the Pearson Chi-square, denoted by 
$\chi^2_{\text{P}}$, when previous experiments suggest an estimate of $\eta$ to 
be $\eta_0$ with standard deviation $\sigma_\eta$,  
one can define the following deviation function:
\begin{eqnarray}\label{eq:chi2def}
& & \chi^2 (\sin^22\theta, |\Delta m^2|, \eta; x) \nonumber \\
&=& 
\chi^2_{\text{P}}(\sin^22\theta, |\Delta m^2|, \eta; x)  + \chi^2_p(\eta) \nonumber \\
&=& \sum_i \frac{(N_i-\lambda_{i}(\sin^22\theta, |\Delta m^2|, \eta))^2}{\lambda_{i}(\sin^22\theta, |\Delta m^2|, \eta)}
+ \frac{(\eta-\eta_0)^2}{(\sigma_\eta)^2}\,.  
\end{eqnarray}
Below, we use the notation $\arg\min_{w} h(w)$ to denote the value of $w$ that 
minimizes any given function $h$, and the standard set-builder notation 
$\{h(w): \text{restriction $w$}\}$ to denote a set that is made up of all the 
points $h(w)$ such that $w$ satisfies the restriction to the right of the 
colon. Let 
\begin{equation}
(\sin^22 \theta_{\min}, |\Delta m^2_{\min}|,\eta_{\min}) = 
\arg\min \chi^2(\sin^22\theta, |\Delta m^2|,\eta; x),
\end{equation} 
that is, the value of $\left({\sin^22\theta, |\Delta m^2|, \eta}\right)\in\Theta \times M \times \Eta$ that best fits the data according to the 
deviation $\chi^2$. 
Also, let 
\begin{equation}
\chi^2_{\min}(x)=\chi^2(\sin^22\theta_{\min},|\Delta m^2_{\min}|,\eta_{\min}; x).
\end{equation} 
And for any given 
$(\sin^22\theta,|\Delta m^2|)$, let 
\begin{equation}
\eta_{\min}(\sin^22\theta,|\Delta m^2|)=\arg \min_{\eta} \chi^2(\sin^22\theta,|\Delta m^2|,\eta; x).
\end{equation}  
Then we can define a test statistic that reflects how much worse $(\sin^22\theta,|\Delta m^2|)$ is than that of the best fit, namely,
\begin{eqnarray}\label{eq:chisqmin}
&\Delta\chi^2 (\sin^22\theta,|\Delta m^2|; x) \equiv&  \nonumber \\
&\chi^2(\sin^22\theta,|\Delta m^2|,\eta_{\min}(\sin^22\theta,|\Delta m^2|); x)-\chi^2_{\min}(x)\,&.
\end{eqnarray}
The corresponding CI with confidence level $\al$ is 
defined to be
\begin{widetext}
\begin{equation}\label{eq:CI}
CI_\al=\{(\sin^22\theta,|\Delta m^2|) 
\in \Theta \times M: \Delta\chi^2(\sin^22\theta,|\Delta m^2|; x) \leq t_\al \}\,.
\end{equation}
\end{widetext}
The term $t_\al$ represents the threshold value such that,  
$\min_\eta \Prob_{\sin^22\theta,|\Delta m^2|,\eta}\left(\Delta \chi^2 \leq t_\al\right) \geq \al$.  
The key in constructing a CI is to specify $t_\al$ correctly for a given 
confidence level $\al$.
  
   Most commonly examined confidence levels use $\al= 68.3\%$
  $(1\sigma)$, $95.5\%$ $(2\sigma)$, $99.7\%$ $(3\sigma)$, which are
  often linked to threshold values $t_\al=$ 2.31, 5.99, 11.8, respectively~\cite{PDG}. 
  Note that these three values are the $68.3\%$, $95.5\%$ and $99.7\%$ 
  quantiles of the Chi-square distribution with two degrees of freedom, 
  respectively. The reason why these threshold values are used is that,
  $CI_\al$ is indeed constructed upon screening the entire parameter
  space by inspecting one point at a time, denoted by $(\sin^22\theta_1,|\Delta m^2_1|)$,
   and testing the pair of hypotheses, 
$H_0: (\sin^22\theta,|\Delta m^2|)=(\sin^22\theta_1,|\Delta m^2_1|)$ 
  versus $H_1:$ otherwise. To test the above hypotheses using the Chi-square 
  statistic $\Delta\chi^2$ in Eq.~\eqref{eq:chisqmin}, the full parameter 
space for $(\sin^22\theta,|\Delta m^2|, \eta)$ is 
  $
  \Theta \times M \times \Eta$, and the null hypothesis space is $
  \{(\sin^22\theta_1,|\Delta m^2_1|)\} \times \Eta$.  
  According to the Wilks' theorem~\cite{Wilks-1938}, 
  if certain regularity conditions hold, mainly 
  \begin{enumerate}
   \item[C1.] the full parameter space $
   \Theta \times M \times \Eta$ is a continuous space, 
	and the the model likelihood function is a smooth function 
	(for example three times differentiable) in the parameters,  
  \item[C2.] 
  the full parameter space contains an open neighborhood around the true value $(\sin^22\theta_{1},|\Delta m^2_{1}|,\eta_1)$, and
  \item[C3.] the data size $N_i$ is large for each $i=1,\ldots, n$,
\end{enumerate} then the statistic 
  $\Delta\chi^2 (\sin^22\theta_1,|\Delta m^2_1|; X)$ 
  follows approximately a Chi-square distribution when 
  $X$ is data generated from $H_0$. 
  Further, 
   the degree of freedom of this Chi-square distribution equals the difference 
   between the dimension of  
   the full parameter space and that of the null hypothesis space, namely $2$, in the current case.
  This procedure of constructing CIs and its extensions have been successfully applied in 
  many studies in order to constrain various parameters in the field of neutrino physics (e.g. 
Ref.~\cite{An:2013zwz}).

\begin{figure*}
\centering
\includegraphics[width=150mm]{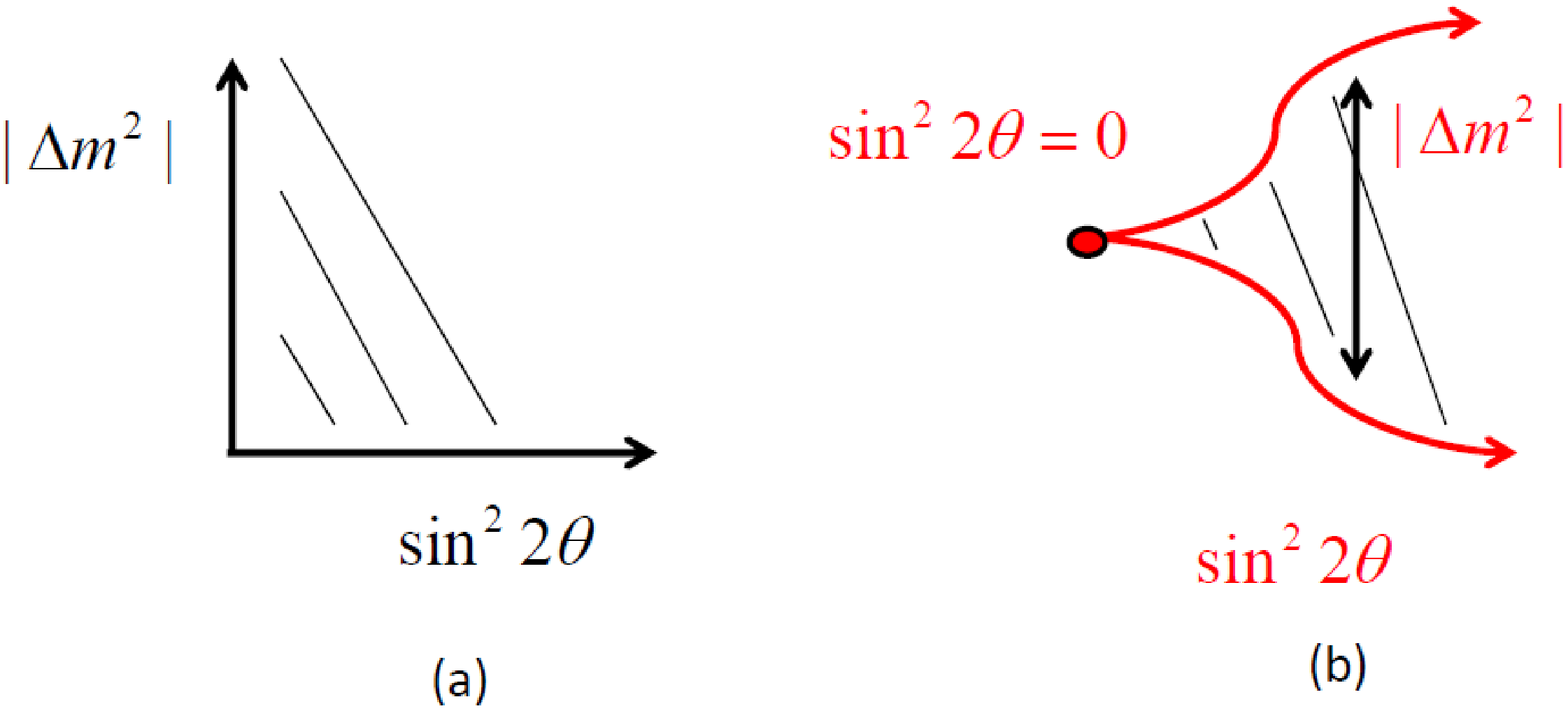}
\caption{(color online) Left panel (a): The parameter space of $\sin^22\theta$ 
	vs. $|\Delta m^2|$ in the Cartesian coordinate. Physical constraints 
	are $\sin^22\theta  \ge 0$ and $|\Delta m^2| \ge 0$.
	Right panel (b): Schematic illustration of the effective parameter space 
	of $\sin^2\theta$ vs. $|\Delta m^2|$ taking into account the spectral 
	difference measured by $\chi^2$ defined in Eq.~\eqref{eq:chi2def}. 
	When $\sin^22\theta = 0$, points with different values of 
	$|\Delta m^2|$ will converge into a single point. This can be easily 
	seen from Eq.~\eqref{eq:osc}. At $\sin^22\theta = 0$, $|\Delta m^2|$ 
	has no impact on the neutrino spectrum. 
	Therefore, when $\sin^22\theta = 0$, there is no open 
	neighborhood around the true value, leading to a failure of 
	regularity conditions required by the Wilks' theorem. 
}
\label{fig:illustration}
\end{figure*}

Although the above Wilks' CI method has been widely used in analyzing experimental 
data, it does not always produce CIs that have correct coverage. 
Its limitations have been addressed by, for example, Feldman and Cousins~\cite{Feldman:1997qc}. 
One example is the searches for neutrino oscillations in the disappearance mode. 
The oscillation probability with $(\sin^22\theta,|\Delta m^2|)$ in a 2-flavor framework 
 is written as: 
\begin{equation}\label{eq:osc}
P_i = 1 - \sin^2 2\theta \cdot \sin^2 (1.27\cdot |\Delta m^2| \cdot L/E^{\nu}_i ), 
\end{equation}
where $L$ and $E^{\nu}_i$ are the distance neutrino travels and the
neutrino energy at the $i$-th bin, respectively. Then the mean bin counts $\lambda_i=\E(\Ni)$ 
are such that $\lambda_i = m\cdot (a_i \cdot P_i + b_i)$, where $a_i$ and $b_i$ are coefficients that depend on 
the vector of nuisance parameters $\eta$, and $m$ represents the amount of 
accumulated data (e.g. the elapsed time for data collection).

The reason why the Wilks' CI method fails for the above neutrino oscillations example 
is the following. 
A key middle step in the proof of the Wilks' theorem is that conditions C1$-$3
together ensure that the estimator of $(\sin^2 2\theta,|\Delta m^2|)$ based on minimizing $\chi^2$, 
has a distribution close to a Gaussian distribution. This suggests two cases. 
(1) When testing a hypothesis $H$ of the form: $\sin^2 2\theta =0$ for any value 
of $|\Delta m^2|$,
C2 is violated, hence the Wilks' theorem does not apply no matter how large the data
 size is. (2) When testing hypotheses of all other forms, C1 and C2 are 
both satisfied, hence as the sample size grows to infinity, the distribution
of $\Delta\chi^2$ will eventually converge to a 
Chi-square distribution. However, for instance if the true $\sin^2 2\theta$ is 
close to $0$, then there could be a non-ignorable probability that we observe a data set 
that results in 
$\sin^22 \theta_{\min}=0$. This clearly prevents the 
distribution of $(\sin^22\theta_{\min},|\Delta m^2_{\min}|)$ 
from being closely approximated by 
a Gaussian distribution. Indeed, the closer $\sin^2 2\theta_{0}$ is 
to $0$, the larger the data size is needed to overcome the above phenomena. 

The latter point can also be understood intuitively. The parameter space 
of $\sin^22\theta$ vs. $|\Delta m^2|$, as is usually displayed in 
Fig.~\ref{fig:illustration}a, is uniform. But the effective parameter space of 
$(\sin^2\theta,|\Delta m^2|)$, in which the distance between any two points is
measured by $\chi^2$ defined in Eq.~\eqref{eq:chi2def}, is no longer uniform (Fig.~\ref{fig:illustration}b).
Due to the functional form of the oscillation formula, the effective
parameter space becomes more compact at smaller $\sin^22\theta$, as the differences
between spectra with different values of $|\Delta m^2|$ become 
smaller. Therefore,  more data is needed to reach the large data limit
required by the Wilks' theorem in order to maintain the 
open neighborhood around the true parameter values (regularity condition C2).
For example, the true $\sin^22\theta = 0$ hypothesis does not have 
an open neighborhood, as $\sin^22\theta<0$ is not allowed. It is therefore impossible 
to reach the large data limit. Even for non-zero but small true value of $\sin^22\theta$, 
the required data size could be well beyond the experimental reach.

When these regularity conditions are not satisfied, there are instances when 
the parent distribution of $\Delta \chi^2$ can have simple approximations that are not necessarily Chi-square. See, for e.g. Ref.~\cite[Sec.~3]{asimov}, where the parameter $\paratheta$ has dimension $1$. For more general cases, one needs the MC method to set CIs. 
Below, we review how to produce a valid 1-$\sigma$
(68\%) CI of ($\sin^22\theta,|\Delta m^2|$) 
using MC, which can be easily generalized to 
build CIs of any level.

Having observed data $x=\{N_1,\cdots, N_n\}$,
apply the following procedure to every 
($\sin^22\theta,|\Delta m^2|$) in the parameter
space $\Theta \times M$: 
\begin{enumerate}
\item Calculate $\Delta \chi^2(\sin^22\theta,|\Delta m^2|;x)$ with
  Eq.~\eqref{eq:chisqmin} based on the observed data.
\item Simulate a large number of MC samples, say $\{x^{(j)}\}_{j=1}^{T}$, 
where $x^{(j)}=\{N_1^{(j)},\cdots,N_n^{(j)}\}$ is generated from the 
model with true parameter value $(\sin^22\theta,|\Delta m^2|)$.  
Here, the nuisance parameters can be either randomly 
generated according to the common hybrid 
Bayesian/Frequentist approach~\cite{Cousins:1991qz} or 
fixed at the best-fit values from data according to the
full Frequentist approach~\cite{full_freq,hybrid_sample}.
For $j=1,\ldots,T$, calculate $\Delta \chi^2(\sin^22\theta,|\Delta m^2|; x^{(j)})$. 
This produces an empirical distribution of the statistic
$\Delta \chi^2$.
\item Calculate the percentage of MC samples such that 
$\Delta \chi^2(\sin^22\theta,|\Delta m^2|; x^{(j)})
<\Delta \chi^2(\sin^22\theta,|\Delta m^2|; x)$. 
Then $(\sin^22\theta,|\Delta m^2|)$ is included in the 1-$\sigma$ 
CI if and only if the percentage is smaller than 68\%. 
\end{enumerate}
The key of the above procedure is to generate an empirical 
distribution of $\Delta \chi^2$, which is not necessarily close to a Chi-square distribution. 

Unlike the Wilks' CI method, the MC CI method 
guarantees the validity of the resulting CIs when the MC sample 
size is large. However,   
the procedure can be very time-consuming when the dimension 
of the vector of unknown parameters is high and/or when a fine grid of the parameter 
space needs to be examined. In addition, the number of MC samples needed to produce an empirical distribution that leads to an accurate enough CI  
increases quickly as the required confidence level increases. 
The procedure can become prohibitively expensive if the minimization process used to find 
($\sin^22\theta_{\min}, |\Delta m^2_{\min}|$) is slow due to the existence
of many nuisance parameters or other technical difficulties. 

Furthermore, there is no simple recipe to strictly combine the CIs generated with the 
MC CI method from different experiments to form an overall CI. 
To see this, consider an example where
several experiments are carried out to probe the parameter space of $(\sin^22\theta, |\Delta m^2|)$. 
For any space point $(\sin^22\theta, |\Delta m^2|)$, the $\Delta \chi^2$ 
statistic of the $j$th experiment is given by $\chi^2(\sin^22\theta,|\Delta m^2|;x^{(j)}) - \chi^2
(\sin^22\theta_{\min}^{(j)},|\Delta m^2_{\min}|^{(j)})$. 
Note that the minimum-value parameter space point, $(\sin^22\theta_{\min},|\Delta m^2_{\min}|)$, 
based on different experiments are typically different. 
Once the experiments are combined, a strict implementation of the MC CI method 
requires to know the global minimum-value parameter space point, which is in general unattainable.  
Indeed, one has to redo MC simulations 
for the combined data, which is expensive in computation since minimization 
has to be done for each MC sample.


In the next section, we introduce a different test statistic from $\Delta\chi^2$, 
which allows for a simple approximation to its distribution under mild conditions. 
Using this new test statistic helps circumvent the computational problems mentioned above.

\section{ the $\Delta T$ Statistic}\label{sec:DeltaT}

\subsection{ Non-nested hypotheses testing}\label{subsec:DT}

Recall that we used $\paratheta$ and $\eta$ to denote the parameter of interest and 
the nuisance parameter respectively. The corresponding model has mean bin counts 
$\left(\lambda_1(\paratheta,\eta), \cdots, \lambda_n(\paratheta,\eta)\right)$. Let $\PT$ denote the parameter space for $\paratheta$.  
In this section, we consider pairs of non-nested hypotheses $H_0: \paratheta=\paratheta_0$ 
and $H_1: \paratheta=\paratheta_1$, one pair at a time, for any 
$\paratheta_0 \neq \paratheta_1 \in \PT$. For convenience and clarity, we 
update some of our notations and refer to the nuisance parameter under $H_0$ and $H_1$ as 
$\nutn$ and $\nuta$ respectively, and they can be of different dimensions. Also, we refer 
to the mean bin counts associated with $\paratheta_0$ and $\paratheta_1$ as $\mu$ and 
$\nu$ respectively, that is, the mean count of the $i$th bin is 
$\mu_i(\nutn)=\lambda_i(\paratheta_0,\nutn)$ under $H_0$ and  
$\nu_i(\nuta)=\lambda_i(\paratheta_1,\nuta)$ under $H_1$.

We now introduce a test statistic, 
denoted by $\Delta T(\paratheta_0,\paratheta_1; x)$, or simply $\Delta T(x)$,
for testing $H_0$ versus  $H_1$.
More than one version of the definition of $\Delta T$ 
will be listed below.

We start with either the Poisson or 
the Normal distribution to model the data $x$, and use the general notation 
$L(x,\lambda)$ to denote the corresponding likelihood, where $\lambda$ equals
to $\mu(\nutn)$ under $H_0$, and  $\nu(\nuta)$ under $H_1$, respectively.
Following the practice of Ref.~\cite[sec.~2]{Baker:1983tu}, 
we convert~\footnote{This is usually done in order that $T_{H_0}(\nutn;x)$ and 
$T_{H_1}(\nutn;x)$ are asymptotically equivalent under certain conditions to their 
counterparts in the classical Chi-square forms, namely, the Neyman and the Pearson 
Chi-square statistics.} the likelihood functions under $H_0$ and $H_1$ into 
$T_{H_0}(\nutn;x)$ and $T_{H_1}(\nutn;x)$ respectively. Let
\begin{equation}
a(x)=2\log L(x,\lambda=x),
\end{equation} and define
\begin{eqnarray}
T_{H_0}(\nutn;x) &=& -2 \log (L(x,\mu(\nutn))) + a(x),\;\;\text{and} \\
T_{H_1}(\nuta;x) &=& -2 \log (L(x,\nu(\nuta))) + a(x)\,,
\end{eqnarray}
both of which can be interpreted as likelihood ratios.
Take the Poisson model for example, we have 
\begin{eqnarray}
2\log L(x,\lambda)&=&\sum_{i=1}^n -2N_i\log\lambda_i +2\lambda_i+2\log(N_i!) \nonumber \\
&\approx & \sum_{i=1}^n [2(\lambda_i -N_i+N_i\log(N_i/\lambda_i))\nonumber \\
&+&\log N_i]+n\log(2\pi),
\end{eqnarray}
 and $a(x)=\sum_{i=1}^n\log N_i +n\log(2\pi)$. Then, looking at the 
definition of $T_{H_0}$ for instance, we have
\begin{equation}
T_{H_0}(\nutn; x) =\sum_{i=1}^{n} 2\left(\mu_i(\nutn) - N_i + N_i \log(N_i/\mu_i(\nutn)) \right).
\end{equation}

In practice, when there are prior experiments carried out to study the 
nuisance parameters, an additional term that reflects deviation from this 
prior knowledge is added to the definition of $T_{H_0}(\nutn; x)$. We denote 
this term by $\chi^2_\text{penalty}(\nutn)$, an example of 
which is the term $\frac{(\eta-\eta_0)^2}{(\sigma_\eta)^2}$ in 
Eq.~\eqref{eq:chi2def}. And when the data size is large, terms of smaller 
order are sometimes omitted from the definition of $T_{H_0}(\nutn;x)$. 
There are at least four common variations for $T_{H_0}(\nutn; x)$ 
used in practice:
\begin{widetext}
\begin{eqnarray}
T_{H_0}(\nutn; x) &=& \sum_{i=1}^{n} 2\left(\mu_i - N_i + N_i \log\frac{N_i}{\mu_i} \right) + \chi^2_\text{penalty}(\nutn), \label{eq:poisson_T}\\
T_{H_0}(\nutn; x) &=& \sum_{i=1}^{n} \log \frac{\mu_i}{N_i} + \sum_{i=1}^{n} \frac{(N_i-\mu_i)^2}{\mu_i} + \chi^2_\text{penalty}(\nutn), \label{eq:Gaussian_T}\\
T_{H_0}(\nutn; x) &=& \sum_{i=1}^{n} \frac{(N_i-\mu_i)^2}{\mu_i} + \chi^2_\text{penalty}(\nutn),  \label{eq:Pearson_T}\\
T_{H_0}(\nutn; x) &=& \sum_{i=1}^{n} \frac{(N_i-\mu_i)^2}{N_i} + \chi^2_\text{penalty}(\nutn). \label{eq:Neyman_T}
\end{eqnarray}
\end{widetext} 
Here, Eq.~\eqref{eq:poisson_T} and Eq.~\eqref{eq:Gaussian_T} origin from 
the likelihood function of the Poisson and the Gaussian distribution, 
respectively. Eq.~\eqref{eq:Pearson_T} and Eq.~\eqref{eq:Neyman_T} are 
variations of Eq.~\eqref{eq:Gaussian_T}, and are commonly 
referred to as the Pearson and the Neyman Chi-square, respectively~\footnote{
As summarized in Ref.~\cite{Baker:1983tu}, all the above estimators had a set of 
properties which the authors considered optimal. They called them "best asymptotically 
normal" (BAN) estimators. The versions of test statistics based directly on 
likelihood functions are considered superior due to their faster convergence 
to the limiting chi-square distributions.} Note that
Eq.~\eqref{eq:chi2def} is a specific example of 
Eq.~\eqref{eq:Pearson_T}. 
We can define four versions of $T_{H_1}(\nuta; x)$ similarly.

En route to form the test statistic $\Delta T$, 
$T_{H_0}(\nutn; x)$ and  $T_{H_1}(\nuta; x)$ are further minimized over all nuisance parameters,
to obtain $T^{\min}_{H_0}(x)=\min_{\nutn} T_{H_0}(\nutn; x)$ and 
$T^{\min}_{H_1}(x)=\min_{\nuta} T_{H_1}(\nuta; x)$, respectively. 
Finally, we define the test statistic 
\begin{equation}\label{eq:deltaT}
\Delta T(x) = T^{\min}_{H_1}(x) - T^{\min}_{H_0}(x)\,.
\end{equation}
Note that $\Delta T(x)$ has the interpretation of being a log-likelihood ratio 
test statistic (or certain variations of it, depending on which version of 
the definition of  $T_{H_0}$ and  $T_{H_1}$ are used) between the two hypotheses. 
It is easy to see that a positive $\Delta T(x)$ would favor 
$H_0$, and a negative $\Delta T(x)$ would favor $H_1$.
In addition, the absolute size of $\Delta T(x)$ reflects how much one 
hypothesis is favored over the other.~\footnote{An alternative 
way to define $\Delta T(x)$ is to replace $T^{\min}_{H_0}(x)$ 
and $T^{\min}_{H_1}(x)$ by $T^{\text{mag}}_{H_0}(x)$ and 
$T^{\text{mag}}_{H_1}(x)$, respectively, which are the marginalized, 
or say integrated version of $T_{H_0}(\nutn; x)$ 
and $T_{H_1}(\nutn; x)$ over all nuisance parameters. 
These two methods generally give very similar results
in practice. From the statistics point of view, while the minimization 
method adopts the Frequentist's philosophy, the marginalization method 
adopts the Bayesian philosophy. }

{\bf Remark.} We emphasize that $\Delta T(x)$ is a different type of test 
statistic than $\Delta \chi^2(x)$  in Eq.~\eqref{eq:chisqmin}. Specifically, 
$\Delta T(x)$ involves the best fit under the restrictions $H_0$ and $H_1$, 
respectively, while $\Delta \chi^2(x)$ involves the best fit under the 
restrictions $H_0$ and over the full parameter space, respectively. 
The way $\Delta T$ is defined is key to why there is a Gaussian approximation 
that works under very general setups, even in the cases where simple 
approximations for the conventional $\Delta\chi^2$ statistic fails. 
Nevertheless, we should note that, when the computing is affordable, forming 
CIs using $\Delta \chi^2$ is more desirable because it leads to a unified 
approach for setting limits 
in absence of new physics signals and in estimating parameters after 
the discovery of new physics~\cite{Feldman:1997qc}. 

Next, we introduce the concept of {\it the Asimov data set}~\cite{asimov}.  Let $x^{\text{Asimov}}_{H_0}$
denote the Asimov data set under $H_0$, which is, loosely speaking,
the mean counts corresponding to the true model in $H_0$ without any
statistical fluctuation nor variations of systematics (change in nuisance 
parameters from their true value).  
In mathematical symbols, $x^{\text{Asimov}}_{H_0}=\mu(\eta^0)$, where $\eta^0$ stands for 
the true value of the nuisance parameter. In practice, we do not know $\eta^0$,
  so it is commonly approximated by 
 an existing nominal value of
  the nuisance parameter (such as the term $\eta_0$ in
  Eq.~\eqref{eq:chi2def}). 
  
Finally, we define a term that will help describe the distribution of the test 
statistic $\Delta T$ under $H_0$. 
Assuming that $H_0$ is the correct hypothesis, define \begin{eqnarray}\label{eq:Tbarn}
\Tbarn &=& \Delta T(x^{\text{Asimov}}_{H_0}) \nonumber\\
	&=& T^{\min}_{H_1}(x^{\text{Asimov}}_{H_0}) - 
   T^{\min}_{H_0}(x^{\text{Asimov}}_{H_0}) \nonumber\\
&=& T^{\min}_{H_1}(x^{\text{Asimov}}_{H_0})\,,
\end{eqnarray}  
where the last step holds because $T^{\min}_{H_0}(x^{\text{Asimov}}_{H_0}) =0$ 
by the definition of $T^{\min}_{H_0}$ and $x^{\text{Asimov}}_{H_0} := \mu(\eta^0)$.

Analogously, let $x^{\text{Asimov}}_{H_1}=\nu(\nuta^0)$ denote the Asimov data set under $H_1$, 
where we can approximate $\nuta^0$ by an existing nominal value. 
Then the following term will help describe the distribution of 
the test statistic $\Delta T$, had $H_1$ been the correct hypothesis:
\begin{eqnarray} \label{eq:Tbara}
     \Tbara &=& \Delta T(x^{\text{Asimov}}_{H_1}) \nonumber \\
     &=& T^{\min}_{H_1}(x^{\text{Asimov}}_{H_1}) -
     T^{\min}_{H_0}(x^{\text{Asimov}}_{H_1}) \nonumber \\
     &=& - T^{\min}_{H_0}(x^{\text{Asimov}}_{H_1}).
 \end{eqnarray}

\subsection{A Gaussian Approximation to the Distribution of $\Delta
  T(X)$ with Large Data Size}\label{sec:Gaussian}
In this section, we show that by omitting terms of relatively small orders, the
distribution of $\Delta T(X)$ 
under hypothesis $H$ follows approximately a Gaussian
distribution with mean $\overline{\Delta T_H}$ and standard deviation
$2\sqrt{|\overline{\Delta T_H}|}$, where $H$ could be either $H_0$ or $H_1$. 

\subsubsection{Description of the mathematical problem and notations}
Recall that we defined four versions of $\left( T_{H_0}(\nutn; x), T_{H_1}(\nuta; x)\right)$ that yield four different definitions of the test statistic $\Delta T(X)$. In this section, we focus on 
studying $\Delta T(X)$ based on Eq.~\eqref{eq:Pearson_T},  namely the Pearson Chi-square 
statistic. For clarity, we call it $D(X)$ from here on. The main part of Sec.~\ref{sec:approxT} 
will be devoted to develop a Gaussian approximation for the distribution of $D(X)$ under $H_0$. 
And in the remarks in the end of Sec.~\ref{sec:approxT}, we show that under $H_0$, the differences between 
the other three versions of $\Delta T(X)$ to $D(X)$ are  insignificant under fairly general conditions, so the 
approximate distribution derived for $D(X)$ can
also be used for all the different versions of $\Delta T(X)$. 
Note that due to the symmetry between $H_0$ and $H_1$, the aforementioned result also 
applies to $D(X)$ and its variations under $H_1$. 

The mathematical problem concerning $D(X)$ is the following. Let
\begin{eqnarray}
\chi^2_{\NH}(\nutn; X)&:=& \sum_{i=1}^\n \frac{({\Ni}-\mu_i(\nutn))^2}{\mu_i(\nutn)} \;\; \text{and} \nonumber \\ 
\chi^2_{\IH}(\nuta; X)&:=& \sum_{i=1}^\n \frac{({\Ni}-\nu_i(\nuta))^2}{\nu_i(\nuta)}
\end{eqnarray}
and let 
\begin{eqnarray}
\hnutn=\arg\min_\nutn \chi^2_{\NH}(\nutn; X) \nonumber \\
\hnuta=\arg\min_\nuta \chi^2_{\IH}(\nuta; X).
\end{eqnarray}
Then the definition of $D(X)$ is
\begin{equation}\label{eq:gaus}
D(X)=\chi^2_{\IH}(X,\hnuta)-\chi^2_{\NH}(X,\hnutn)\,.
\end{equation}
Note that $\left(\chi^2_{\NH}(X,\hnutn), \chi^2_{\IH}(X,\hnuta)\right)$ stands for the version of $\left(T^{\min}_{H_0}(X), T^{\min}_{H_1}(X)\right)$ that is based on Eq.~\eqref{eq:Pearson_T}. 

Our goal is to obtain an approximation of the distribution of $D(X)$ under $H_0$, when 
the data size is large.
 Hence a specific quantity, say
$m$, is needed to reflect the magnitude of the data, in order that
we can describe how other quantities in the model change along with
it.  For example, $m$ could be the duration of the experiment or the
total number of events. For the ease of description, let $m$ represent the
duration of the experiment in this section.  Then $\p=\frac{X}{m}$
stands for the per unit time observed counts in a potential
experiment, and it would remain stable (instead of tending to infinity
or zero) as $m$ grows.  So we say $\p$ is of order $O_p(1)$ (with
respect to $m$)
\footnote{The Big $O$, the small $o$, the Big $O_p$, and the
  small $o_p$ notation are standard mathematical symbols, such that for two 
sequences of random variables $\{X_m\}$ and $\{Y_m\}$, we write
\begin{itemize}
\item $Y_m=o_p(X_m)$ if and only if $Y_m/X_m \rightarrow 0$ in
  probability as $m \rightarrow \infty$, and
\item $Y_m=O_p(X_m)$ if and only if $Y_m/X_m$ is bounded in
  probability as $m \rightarrow \infty$\,.
\end{itemize}
In the special case where $\{X_m\}$ and $\{Y_m\}$ are deterministic sequences, the 
stochastic $o_p$ and $O_p$ symbols reduce to the $o$ and $O$ symbols. 
See Ref.~\cite[sec 2.2]{van:1998} for details on the rules of calculus with these symbols.
}.

In order to describe the modeling of counts rigorously, we introduce a
set of notations, a summary of which is provided in
Table~\ref{tab:legend}.  Recall that when $H_0$ is the correct hypothesis, we employed $\mu(\nutn)$ to denote the mean bin counts for models under this hypothesis, where $\nutn$ is the vector of unknown nuisance parameters of
dimension $q$.  Denote the true value of $\nutn$ by $\nutnT$, that is,
$\muT=\mu(\nutnT)$ is the true mean counts of the observation such
that
\[{\Ni} \ind\;\; \text{Poisson}(\muT_i)\;\;\;\text{for
  $i=1,\cdots,\n$.}\] 
  When the data size is large, a very good
approximation to the model above is given by
\[{\Ni} \ind\;\; \text{N}(\muT_i,\muT_i)\;\;\;\text{for
  $i=1,\cdots,\n$.}\]
Further, let $\pi:=\mu/m$ 
denote the per unit time mean counts.  To help explain these
notations, take the example from Sec.~\ref{sec:procedure} for
instance, if $H_0: (\sin^2 2\theta,|\Delta m^2|)=(\sin^22\theta_0, |\Delta
m^2_{0}|)$ is the correct hypothesis, then $\pi_i=\mu_i(\nutn)/m= a_i(\nutn) \cdot P_i + b_i(\nutn)$. 
The terms $a_i$ and $b_i$ are 
functions of order $O(1)$, and are determined by the configuration of the
experiment. For example, $a_i$ can represent the detector efficiency, neutrino flux
from reactor, target mass, etc., $b_i$ can represent the backgrounds. Also,
$P_i = 1 - \sin^2 2\theta_0 \cdot
\sin^2 (1.27\cdot |\Delta m_0^2| \cdot L/E^{\nu}_i )$ represents the
survival probability in a disappearance model.

\begin{table*}
  \begin{center}
    \begin{tabular}{|l|l|l|}
      \hline
      \rule{0pt}{2.4ex} & \textbf{Under the correct hypothesis} & \textbf{Under the alternative hypothesis}  \\
      \hline
      \hline
      \textbf{General notation}  & &\\
      \hline
      Mean bin counts &   $\mu(\nutn)=(\mu_1(\nutn),\cdots,\mu_N(\nutn))$ &   $\nu(\nuta)=(\mu_1(\nuta),\cdots,\mu_N(\nuta))$\\
      \hline
      Per-unit mean counts &  $\pi(\nutn)=\mu(\nutn)/m$ &   $\tau(\nuta)=\nu(\nuta)/m$\\
      \hline
      \hline
      \textbf{True values or their closest approximations}  & &\\
      \textbf{under the give model}  & &\\
      \hline
      nuisance parameter &  $\nutn_0$ (a $q$-dim vector) &   $\nuta_0$ (a $q^*$-dim vector)\\
      \hline
      Mean bin counts &  $\mu^0=\mu(\nutn_0)$ &   $\nu^0=\nu(\nuta_0)$\\
      \hline
      Per-unit mean counts &  $\pi^0=\mu^0/m$ & $\tau^0=\nu^0/m$ \\
      \hline
      \hline
      \textbf{Estimation based on observed data}  & &\\
      \hline
      nuisance parameter &  $\hnutn=\arg\min\chi^2_{\NH}(\nutn; X) $ &   $\hnuta=\arg\min\chi^2_{\IH}(\nuta; X)$\\
      \hline
      Mean bin counts &  $\hmu=\mu(\hnutn)$ &   $\hnu=\nu(\hnuta)$\\
      \hline
      Per-unit mean counts &  $\hpi=\hmu/m$ & $\htau=\hnu/m$ \\
      \hline
    \end{tabular}
    \caption{Legend of symbols used in describing the correct model and the alternative model, respectively.}
    \label{tab:legend}
  \end{center}
\end{table*}

Meanwhile, a competing framework, namely 
the collection of models that satisfy $H_1$,
specifies the mean counts incorrectly as $\nu(\nuta)$, where
$\nuta$ is the unknown nuisance parameter of dimension $q^*$.  Also,
define the per unit time mean counts under $H_1$ by
$\tau=\nu/m$. When $H_0$ is the correct hypothesis and that the true model is $\muT$,
 there exists a unique $\nutaT$, such that
$\hnuta$ approaches 
$\nutaT$ as $m\rightarrow \infty$. We will show in Appendix~A that
$\nutaT$ has the interpretation that it corresponds to the model
$\nu(\nuta)$ among all that belong to the alternative framework that
is the closest to the true model $\muT$ in terms of the deviation
$\sum_{i=1}^\n\frac{(\mu^0_i-\nu_i(\nuta))^2}{\nu_i(\nuta)}\,$.
Denote $\nu^0=\nu(\nutaT)$.

\begin{widetext}
\subsubsection{Approximating the distribution of the test statistic $D(X)$}\label{sec:approxT}

In this section, we always assume that $H_0$ is the correct hypothesis, under which we study the distribution of $D(X)$ defined in Eq.~\eqref{eq:gaus}. For convenience, we will suppress the dependence on
$X$ in the notation, and write
$D=\chi^2_{\IH}(\hnuta)-\chi^2_{\NH}(\hnutn)$. On one hand, it's well
known that the distribution of $\chi^2_{\NH}(\hnutn)$ approaches the
Chi-square distribution with degree of freedom $(\n-q)$ as $m$
increases. On the other hand, the limiting distribution of
$\chi^2_{\IH}(\hnuta)$ as $m$ increases does not always exist.
Indeed, the behavior of $\chi^2_{\IH}(\hnuta)$ for large $m$ is
dependent on how far apart the mean counts of the best model under
the alternative theoretical frameworks are from that of the true
model. Denote the difference of per unit mean counts between the
two models by $\delta=\pi^0-\tau^0$. First, we state a classical assumption made in many 
statistical literatures (such as Ref.~\citep{wald:1943} and Ref.~\cite{asimov}) in order to obtain the limiting distribution of test statistics 
analogous to
$\chi^2_{\IH}(\hnuta)$, that is, the different versions of $T^{\min}_{H_1}(X)$:\\\vspace{-2mm}

\begin{center}
$[$A1$]$ 
$\delta=\pi^0-\tau^0= O(m^{-\frac{1}{2}})$, that is, $\mu^0-\nu^0=O(m^{\frac{1}{2}})$.\\
\end{center}
Assumption [A1] means that the best model under the wrong hypothesis is just barely incorrect. 
For example, under [A1], Ref.~\citep{wald:1943} showed that the likelihood ratio test statistic for 
testing ${\IH}$ against the full parameter space (that is, the statistic $T^{\min}_{H_1}(X)$ based 
on Eq.~\eqref{eq:poisson_T} and Eq.~\eqref{eq:Gaussian_T}) has a limiting non-central Chi-square 
distribution. 
The non-centrality parameter has the same form as the test statistic, but with $\hat{\mu}$ 
and $\hat{\nu}$ replaced by $\mu^0$ and $\nu^0$. For a simplified presentation of this result, 
see, for example, Ref.~\citep[Sec.~3.1]{asimov}. Under [A1], the non-centrality parameter is finite 
and the non-central Chi-square approximations are accurate to the extent that the 
$O_p(m^{-\frac{1}{2}})$ terms are neglected.

 In contrast to [A1], we consider the following assumption, which is more general and realistic 
for the physics problem at hand:\\\vspace{-2mm}

\begin{center}
$[$A0$]$ 
$\delta=\pi^0-\tau^0=O(1)$, that is, $\mu^0-\nu^0=O(m)$.\\
\end{center}
In words, [A0] assumes that the difference in mean bin counts between the best model under the 
wrong hypothesis and the true model increases at the same rate as the data size $m$, or slower. 
Clearly, [A0] is a more relaxed condition than [A1], in the sense that [A1] implies [A0], but not vice versa. An example where  [A0] holds and [A1] 
does not, is the case that the nuisance parameter is absent: each hypothesis allows exactly one model, such that the per unit mean bin counts of the model under $H_0$ is $\pi^0$, and that under $H_1$ is $\tau^0$, where $\pi^0$ and $\tau^0$ are vectors of constants that do not change with the data size $m$. As for general cases where there are nontrivial nuisance parameters, it is possible that the best model under $H_1$ can lead to $\tau^0$ values that move closer to the truth $\pi^0$ as more data become available. Hence $[A1]$ may become satisfied, while [A0] is always satisfied. In situations where one is unwilling to assert a convergence rate as fast as $O(m^{-\frac{1}{2}})$ for $(\tau^0-\pi^0)$, if the convergence occur at all, [A0] is more appropriate than [A1].

To see the impact of using  [A0] instead of [A1], it turns out that when 
$\lim_{m\rightarrow \infty}m^{\frac{1}{2}}\delta=\infty$, various test statistics similar to 
$\chi^2_{\IH}(\hnuta)$ (these are the different versions of $T_{H_1}^{\min}$ that we mentioned 
in Sec.~\ref{subsec:DT}) would be unbounded in probability. Take the likelihood ratio 
test statistic mentioned above for example, the non-centrality parameter in the previous 
approximation grows to infinity as $m$ increases.
Further, the differences between the different versions of $\chi^2_{\IH}(\hnuta)$ usually do 
not converge to 0 as $m$ increases. 

Although the limiting distribution does not necessarily exist under assumption [A0], it is 
still possible to approximate the distribution of $\chi^2_{\IH}(\hnuta)$ at a finite, but 
large enough $m$. We make such an attempt, but this certainly requires a different derivation 
than the existing proofs that assume [A1]. 
In our derivation, we keep track of the terms that have higher order than constants when the 
data size $m$ grows. 
Our proof follows the lines of that of Ref.~\citep[Chap.~16]{stat3}, but with significant modifications. 

Write
\[\chi^2_{\IH}(\hnuta)=\sum_{i=1}^\n \frac{({\Ni}-\hnu_i)^2}{\hnu_i}=:\sum_i f_i^2\;\;\;\text{and}\;\;\;\;\;\;
\chi^2_{\NH}(\hnutn)=\sum_{i=1}^\n
\frac{({\Ni}-\hmu_i)^2}{\hmu_i}=:\sum_i e_i^2\,.\] Here
\[f_i=f_i(\ba,\bb,\bc)=\frac{{\Ni}-\hnu_i}{\hnu_i^{\frac{1}{2}}}=\sqrt{m}\,\frac{\p_i-\htau_i}{\htau_i^{\frac{1}{2}}}=
\sqrt{m}\,\frac{(\p_i-\pi^0_i)-(\htau_i-\tau^0_i)+(\pi^0_i-\tau^0_i)}{((\htau_i-\tau^0_i)+\tau^0_i)^{\frac{1}{2}}}=:
\sqrt{m}\,\frac{a_i-c_i+\delta_i}{(c_i+\tau^0_i)^{\frac{1}{2}}}\,,
\] 
and
\[e_i=e_i(\ba,\bb,\bc)=\frac{{\Ni}-\hmu_i}{\hmu_i^{\frac{1}{2}}}=\sqrt{m}\,\frac{\p_i-\hpi_i}{\hpi_i^{\frac{1}{2}}}= \sqrt{m}\,\frac{(\p_i-\pi^0_i)-(\hpi_i-\pi^0_i)}{((\hpi_i-\pi^0_i)+\pi^0_i)^{\frac{1}{2}}}=: \sqrt{m}\,\frac{a_i-b_i}{(b_i+\pi^0_i)^{\frac{1}{2}}}\,,
\]  
where $a_i=\p_i-\pi^0_i=O_p(m^{-\frac{1}{2}})$, $b_i=\hpi_i-\pi^0_i$, and $c_i=\htau_i-\tau^0_i=O_p(m^{-\frac{1}{2}})$, for $i=1,\cdots, n$.
Then by the Taylor expansion of $\f=(f_1,\cdots, f_n)^T$ and $\e=(e_1,\cdots, e_n)^T$ around $(\ba, \bb, \bc)=(0,0,0)$, we have
\[\f=\sqrt{m}\, \dg\{\tau^0\}^{-\frac{1}{2}}\delta+ 
\dg\{\tau^0\}^{-\frac{3}{2}}\left[\dg\{\tau^0\} -\frac{1}{2}\,
  \dg\{\pi^0+\tau^0\} E^*\right]\sqrt{m}(\p-\pi^0) +
O_p(m^{-\frac{1}{2}})\,,\] where the three terms in the above
expression are of order $O_p(m^{\frac{1}{2}})$, $O_p(1)$ and
$O(m^{-\frac{1}{2}})$ respectively.  Further,
 \[\e=\dg\{\piT\}^{-\frac{1}{2}}(I-D)\sqrt{m}(\p-\piT)+O_p(m^{-\frac{1}{2}})
 \,,\]
where $D= B (B^T \dg\{\piT\}^{-1} B)^{-1} B^T \dg\{\piT\}^{-1}$\,,
and the two terms in the above expression are of order $O_p(1)$ and $O_p(m^{-\frac{1}{2}})$ respectively. 
Therefore
\[\begin{split}
D=&\chi^2_{\IH}(\hnutn)-\chi^2_{\NH}(\hnuta)=\f^T\f-\e^T\e\\
=& m\, \delta^T\dg\{\tau^0\}^{-1} \delta+2\sqrt{m}\,\delta^T \dg\{\tau^0\}^{-2}\left[\dg\{\tau^0\} -\frac{1}{2}\, \dg\{\pi^0+\tau^0\} E^*\right]\sqrt{m}(\p-\pi^0)+O_p(1)\\
=&m\, \delta^T\dg\{\tau^0\}^{-1} \delta+2\sqrt{m}\,\delta^T \dg\{\tau^0\}^{-1}\sqrt{m}(\p-\pi^0)  \\
&-\frac{1}{2}\, \sqrt{m}\,\left[\delta^T
  \dg\left\{\frac{\pi^0+\tau^0}{(\tau^0)^2}\right\} B^*\right]
\left({B^*}^T \dg\left\{\frac{(\pi^0)^2}{(\tau^0)^3}\right\}
  B^*\right)^{-1}{B^*}^T
\dg\left\{\frac{\pi^0}{(\tau^0)^2}\right\}\sqrt{m}(\p-\pi^0)+O_p(1)
\end{split}\] According to Eq.~\eqref{eq:lemma2} of Lemma 1, the term in
the closed bracket above reduces to $0$. Hence
\[D= m\, \delta^T\dg\{\tau^0\}^{-1} \delta+2\sqrt{m}\,\delta^T \dg\{\tau^0\}^{-1}\sqrt{m}(\p-\pi^0) +O_p(1)\,. \]
Denote the first term of $D$ by
\begin{equation}\label{eq:D}
D_1=m\,
  (\pi^0-\tau^0)^T\dg\{\tau^0\}^{-1}(\pi^0-\tau^0)=\sum_{i=1}^\n\frac{(\mu^0_i-\nu^0_i)^2}{\nu^0_i}=\min_{\nu}\sum_{i=1}^\n\frac{(\mu^0_i-\nu_i)^2}{\nu_i}
  =:\D 
  \,,\end{equation} where the second to last equality follows from
Appendix~A. Note that under assumption [A0], $D_1=\D$ 
is of order $O(m)$.  Next denote the second term of $D$ by $D_2$. The
central limit theorem implies that as $m$ increases to infinity,
$\sqrt{m}\,(\p-\pi^0)$ converges in distribution to the
N$(0,\dg\{\pi^0\})$ distribution. Hence $D_2/(2\sqrt{m})$ converges in
distribution to the N$(M_2,V_2)$ where \[M_2=\delta^T
\dg\{\tau^0\}^{-1} 0=0\,,\] and
\[\begin{split} V_2&=\delta^T \dg\{\tau^0\}^{-1} \dg\{\pi^0\}\dg\{\tau^0\}^{-1} \delta\\
  &=\delta^T \dg\{(\tau^0)^{-1}\} \delta+\delta^T \dg\{\frac{\pi^0-\tau^0}{(\tau^0)^2}\} \delta\\
  &=\frac{\D
}{m}+\sum_{i=1}^\n\frac{(\pi^0_i-\tau^0_i)^3}{(\tau^0_i)^2}=:\frac{\D 
}{m}+s\,.\end{split}\] Note that both $\frac{\D 
}{m}$ and $s$ are of order $O(1)$ under assumption [A0], hence
$D_2=O_p(m^{\frac{1}{2}})$.  In summary, under assumption [A0],
we have $D=D_1+D_2+O_p(1)
$, where 
\begin{equation}\label{eq:proof}
 D_1+D_2 \apsim \text{N} (\D, 4\D+4ms)\,.
\end{equation}

\end{widetext}

\textbf{Remarks and Implications of Eq.~\eqref{eq:proof}} 

\begin{enumerate}
\item For the common physics problem that we are interested in, 
additional simplification can be made to the approximating distribution, 
$\text{N}(\D, 4\D+4ms)$. Specifically, in searching for
new physics through precision measurements, the mean bin counts  of the true model and that of the best model under the alternative hypothesis are relatively close to each other, that is,
 \[|\mu^0_i-\nu^0_i| << \mu^0_i \sim \nu^0_i,\] or in other words,
  \begin{eqnarray}\label{eq:cond}
  \frac{(\mu_i^0 -
    \nu_i^0)}{\mu_i^0}&=&\frac{ \delta^0_i}{\pi^0_i} << 1 \;\;\text{and}\;\; \nonumber\\
 \frac{(\mu_i^0 -
    \nu_i^0)}{\nu_i^0}&=&\frac{ \delta^0_i}{\tau^0_i} << 1\,.
\end{eqnarray}
In such situations, one can 
  ignore the $ms$ term in Eq.~\eqref{eq:proof}, 
  because $ms=\sum_i \frac{(\mu_i^0 -
    \nu_i^0)^2}{\nu_i^0}\cdot\frac{\mu_i^0 - \nu_i^0}{\nu_i^0}<<
  \sum_i \frac{(\mu_i^0 - \nu_i^0)^2}{\nu_i^0} = \D $.  Then
  our main result becomes 
\begin{equation}\label{eq:ultimate}
D(X) \apsim N(\D, 4\D).
\end{equation}

\item We claimed in Sec.~\ref{sec:CLs_1} that, at large data limit, the three versions of 
$\Delta T(X)=T^{\min}_{H_1}(X) - T^{\min}_{H_0}(X)$ based on the definition of $T_{H_0}$ 
(and the corresponding $T_{H_1}$) in Eq.~\eqref{eq:poisson_T}, \eqref{eq:Gaussian_T}, 
and Eq.~\eqref{eq:Neyman_T}, each have negligible difference from the $D(X)$. 
We validate this claim as follows.

For the moment, we drop the penalty term $\chi^2_\text{penalty}(\nutn)$ from 
Eq.~\eqref{eq:poisson_T}--\eqref{eq:Neyman_T} for simplicity. And we will address the 
issue of the penalty term in the next remark.

First, for Eq.~\eqref{eq:poisson_T}, 
we have, 
\begin{eqnarray}
T_{H_0}(X) &=& \sum_i 2\left(\mu_i - N_i + N_i \log \left(1 - \frac{\mu_i-N_i}{\mu_i}\right)\right). \nonumber \\
&=&  \sum_i \frac{(\mu_i-N_i)^2}{\mu_i} + O_p(m^{-\frac{1}{2}}) \nonumber \,.
\end{eqnarray}
The last step was obtained through expanding $\log \left(1 - \frac{\mu_i-N_i}{\mu_i}\right)$ at large data limit,($|N_i - \mu_i| = O_p(\mu_i^{\frac{1}{2}})=O_p(m^{\frac{1}{2}})$).
Next, for Eq.~\eqref{eq:Gaussian_T},  
we have,
\begin{eqnarray*}
T_{H_0}(X) 
&=& \sum_i \left( \frac{(\mu_i-N_i)^2}{\mu_i} + \log \left(1+\frac{\mu_i-N_i}{N_i}\right) \right) \\
&=&\sum_i\frac{(\mu_i-N_i)^2}{\mu_i}+O_p(m^{-\frac{1}{2}}).
\end{eqnarray*}
Finally, for Eq.~\eqref{eq:Neyman_T}, we have,
\begin{eqnarray*}
T_{H_0}(X) &=& \sum_i \frac{(\mu_i-N_i)^2}{N_i} \\
&=& \sum_i\frac{(\mu_i-N_i)^2}{\mu_i}+O_p(m^{-\frac{1}{2}}).
\end{eqnarray*}
The differences between each version of $T_{H_0}(X) $ and 
$ \sum_i\frac{(\mu_i-N_i)^2}{\mu_i}$ are negligible.
Next we examine the differences between each version of $T_{H_1}(X) $ and 
$ \sum_i\frac{(\nu_i-N_i)^2}{\nu_i}$. We will only consider situations where 
condition Eq.~\eqref{eq:cond} hold. If so, the term 
$\frac{\nu_i-N_i}{N_i}=\frac{\delta_i}{\pi_i}+O_p(m^{-\frac{1}{2}})<<1$,
which will help validate the following three approximations.
First, for Eq.~\eqref{eq:poisson_T}, we have, 
\begin{widetext}
\begin{eqnarray*}
T_{H_1}(X) &=& \sum_i 2\left(\nu_i - N_i + N_i \log \frac{N_i}{\nu_i}\right)  
= \sum_i 2\left(\nu_i - N_i + N_i \log \left(1 - \frac{\nu_i-N_i}{\nu_i}\right)\right)  \\
 &=& \sum_i \frac{(\nu_i-N_i)^2}{\nu_i} \cdot \left(1 + O\left(\frac{(\nu_i-N_i)}{\nu_i}\right)\right)\,.
\end{eqnarray*}
Next, for Eq.~\eqref{eq:Gaussian_T},  we have,
\begin{eqnarray*}
T_{H_1}(X) 
= \sum_i  \frac{(\nu_i-N_i)^2}{\nu_i} - \sum_i \log \left(1-\frac{\nu_i-N_i}{\nu_i}\right)
 = \sum_i\frac{(\nu_i-N_i)^2}{\nu_i} \cdot \left(1+ O\left(\frac{1}{\nu_i-N_i}\right)\right)\,.
\end{eqnarray*}

Finally, for Eq.~\eqref{eq:Neyman_T}, we have,
\begin{eqnarray*}
T_{H_1}(X) = \sum_i \frac{(\nu_i-N_i)^2}{N_i}
= \sum_i\frac{(\nu_i-N_i)^2}{\nu_i} \cdot \left(1+ O\left(\frac{\nu_i - N_i}{N_i}\right)\right)\,.
\end{eqnarray*}
\end{widetext}
In situations where Eq.~\eqref{eq:cond} is satisfied, the differences between each version of $T_{H_1}(X) $ and $ \sum_i\frac{(\nu_i-N_i)^2}{\nu_i}$ are very small compared to the latter and are hence negligible.

It follows that the different versions 
of the test statistic $\Delta T(X)$ will behave similarly as $D(X)$. 
Finally, it is easy to see that our definition for $\D$ in Eq.~\eqref{eq:D} is 
equivalent to $\overline{\Delta T}$ (defined in Eq.~\eqref{eq:Tbarn})
based on Eq.~\eqref{eq:Pearson_T}. 
Our main result in Eq.~\eqref{eq:ultimate} can be stated as
\begin{equation}
\Delta T \apsim N(\overline{\Delta T}, 4|\overline{\Delta T}|).
\end{equation}

\item 
  We emphasize that Eq.~\eqref{eq:gaus} is a specific form of T 
 in Eq.~\eqref{eq:poisson_T}, \eqref{eq:Gaussian_T}, 
\eqref{eq:Pearson_T}, and Eq.~\eqref{eq:Neyman_T}. 
  The penalty term in $T$ represents the constraint of systematic
  uncertainties, and is commonly obtained by dedicated
  measurements. When one includes the dedicated measurements as part
  of Chi-square definition, one naturally recovers Eq.~\eqref{eq:gaus}.
  Therefore, our proof in Sec.~\ref{sec:Gaussian} is also valid for
  test statistics with the format of $T$ in Eq.~\eqref{eq:poisson_T}, 
 \eqref{eq:Gaussian_T}, \eqref{eq:Pearson_T}, and Eq.~\eqref{eq:Neyman_T}.\\

\item 
We comment on the large data limit, which is required to reach the final
conclusion (Eq.~\eqref{eq:ultimate}) and to show the equivalence of 
Eq.~\eqref{eq:poisson_T}, \eqref{eq:Gaussian_T}, \eqref{eq:Pearson_T}, and 
Eq.~\eqref{eq:Neyman_T}. For a single bin, 
$\frac{N_i-\mu_i}{\mu_i}$ is negligible if $N_i$ is greater than about $100$. 
For multiple bins, the contributions from each bin will likely cancel 
and the condition can be relaxed in practice to that the total number of 
events, $\sum_i N_i$, is greater than about $100$.
\end{enumerate}

In summary of this section, we showed the following result. Assume the following set 
of conditions hold:
\begin{enumerate}
 \item CD1: the parameter space ($\Eta$) of the nuisance parameters $\nutn$ and $\nuta$
   	 are both continuous and the the model likelihood function is a smooth function 
	(for example three times differentiable) in the parameters,  
 \item CD2: the data size $N_i$ is large for each $i=1,\ldots, n$,
 \item CD3: the best model under the null hypothesis $H_0$ and the alternative hypothesis $H_1$ 
are relatively close, in the sense that 
$|\mu_i^0-\nu_i^0| << \mu_i^0 \sim \nu_i^0$, for $i=1,\ldots, n$\,.
\end{enumerate}
Then a simple approximation for the distribution of $\Delta T(X)$ 
under $H_j$, for either $j=0$ or $1$, is the Gaussian distribution 
with mean $\overline{\Delta T_{H_j}}$ and standard deviations  
$2\sqrt{|\overline{\Delta T_{H_j}}|}$.
Based on the Gaussian approximation, the CL$_s$ value is easily calculated with 
$\Delta T (x)$, $\overline{\Delta T_{H_0}}$, and $\overline{\Delta T_{H_1}}$.  
In case any of the above conditions (CD1-CD3) breaks down, the 
distribution of $\Delta T(X)$ is not necessarily well approximated by the Gaussian 
distribution, and should instead be estimated through Monte Carlo simulations.

\section{The CL$_s$ Approach Based on the $\Delta T$ statistic}
\label{sec:CLs}

The $\Delta T(x)$ statistic described in the previous section can be used to form both 
CIs and CL$_s$, and they differ in how the associated p-values are utilized. Note that 
both procedures are easy to carry out because of the simple Gaussian approximation for 
the distribution of $\Delta T(x)$. We will introduce the CL$_s$ approach with the 
$\Delta T(x)$ statistic below. The principle of forming CIs with $\Delta T(x)$ is the same 
as that with $\Delta \chi^2$. 

\subsection{The CL$_s$ Approach Based on the $\Delta T$ Statistic}\label{sec:CLs_1}

\begin{figure}
\centering
\includegraphics[width=75mm]{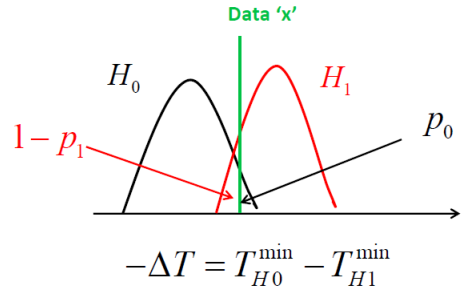}
\caption{(color online) Illustration of the CL$_s$ approach with log-likelihood ratio
$\Delta T$. In order to be consistent with the convention in Ref.~\cite{Lyons:1900zz}, 
we plot the densities of $-\Delta T$ instead. See text for more discussions.}
\label{fig:CLs}
\end{figure}

The CL$_s$ approach~\cite{Read:2000ru,Junk:1999kv,Read:2002hq} is a popular approach 
to present searches for new physics beyond the Standard Model. Recent examples
of using this approach in neutrino physics can be found in 
Ref.~\cite{An:2014bik,Anokhina:2014qda}. Examples of using this approach in 
LHC super particle search can be found at Ref.~\cite{Aad:2012ms,Chatrchyan:2012jx}.
We emphasize that, the CL$_s$ approach is a different way to present statistical results 
than the traditional approach of setting confidence intervals (CI). The traditional CI 
approach is appropriate in treating established signals~\cite{Read:2002hq}.
Whereas the CL$_s$ approach is appropriate 
in setting exclusion limits, such that models with parameter values beyond the limits are much worse than the Standard Model in fitting the observed data. In this 
section, we briefly review the principle of the CL$_s$ approach in a 
two-hypotheses testing problem.

Fig.~\ref{fig:CLs} is a heuristic illustration of the distribution of the log-likelihood 
ratio $\Delta T(X)$, where $X$ stands for data from a potential repeat of the experiment. 
The black (red) curve stands for the density function of the expected distribution of  $\Delta T (X)$ 
under the assumption 
that the null (alternative) hypothesis is true. The green line represents $\Delta T (x)$ 
calculated from the observed data $x$. A positive (negative) $\Delta T(x)$ would favor 
$H_0$ ($H_1$) over $H_1$ ($H_0$). The CL$_s$ value is then defined as:
\begin{equation}\label{eq:cls}
{\rm CL}_s(x) = \frac{1-p_1}{1-p_0},
\end{equation}
where $1-p_1$ ($1-p_0$) is the probability that a potential repeat of the 
experiment will yield a $\Delta T(X)$ value larger than $\Delta T (x)$ 
when the alternate (null) hypothesis is true.   
Hence, the definition of CL$_s$ in Eq.~\eqref{eq:cls} suggests that a 
CL$_s$ value close to zero would favor $H_0$ against $H_1$. On 
the other hand, as illustrated in Fig.~\ref{fig:CLs2}, a CL$_s$ value 
close to one does not necessarily indicate that $H_1$ is favored against $H_0$.
In searching for new physics beyond the Standard Model, $H_0$ is 
typically specified to be the Standard Model. The exclusion region of 
the parameter space is typically defined as the set of parameter values of new physics that 
corresponds to CL$_s$ value smaller than $\alpha=0.05$~\cite{Lyons:1900zz},
while other threshold values of the CL$_s$ can be used as well.

\begin{figure}
\centering
\includegraphics[width=75mm]{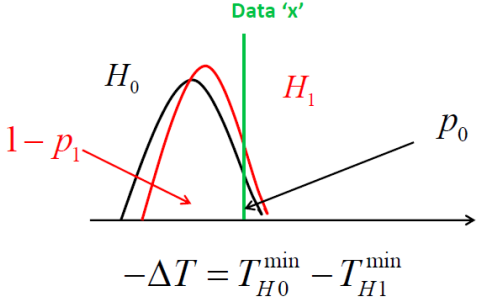}
\caption{(color online) Same as Fig.~\ref{fig:CLs} except that the null 
hypothesis $H_0$ and the alternative hypothesis
$H_1$ are very close to each other. }
\label{fig:CLs2}
\end{figure}

Note that the CL$_s$ value is never smaller than $(1-p_1)$, the p-value used in the corresponding CI approach. Hence, had the exclusion contour at $\alpha$ been used to set a CI, it would have coverage probability 
over $1-\alpha$. Nevertheless, the CL$_s$ value appears to be a more reasonable measure of extremeness 
than $(1-p_1)$, in situations where $H_0$ and $H_1$ are very similar (see Fig.~\ref{fig:CLs2}).
For example, assuming the data $x$ is
an ``extreme'' measurement with respect to  $H_1$ 
(i.e. small $p_1$), it will also be disfavored by  $H_0$ 
(i.e. small $p_0$).   If only a single p-value, either $p_0$ or $p_1$, 
is examined as in the CI approach, then one would draw the inappropriate 
conclusion of excluding $H_0$ or $H_1$ while favoring the other hypothesis.
However, since the hypotheses $H_0$ and $H_1$ are similar, the data 
does not carry enough information to 
differentiate them. The CL$_s$ value, which is the ratio between $1-p_1$ and $1-p_0$ will 
protect against such situations. 

In order to obtain the value of $p_0$ and $p_1$ required to 
calculate the CL$_s$, one needs to find the distribution of $\Delta T(X)$ under $H_0$ and $H_1$. 
While Monte Carlo simulations
can provide approximations to the distribution of $\Delta T(X)$, 
simpler methods, such as Gaussian approximations, are desired to lower the computing burden. 

\subsection{Setting Exclusion Sets with the Gaussian CL$_s$ Method}\label{sec:procedure}

In this section, we illustrate the procedure of setting exclusion 
sets with the Gaussian CL$_s$ method for the neutrino oscillation 
example from Sec.~\ref{sec:review}. 

Here the parameter of interest is $\paratheta=(\sin^2 2\theta,|\Delta m^2|)$. 
The  mean count for the $i$th bin is described as 
$\mu_i(\nutn) = m\cdot (a_i(\nutn) P_i(\sin^22\theta_0,|\Delta m_0^2|) + b_i(\nutn))$,
where $a_i$ and $b_i$ are coefficients that depend on 
the vector of nuisance parameters $\eta$, and $m$ represents the amount of 
accumulated data.
It is typical to use $\paratheta_0=(0, |\Delta m^2_{0}|)$ as a reference parameter point, 
where $|\Delta m^2_{0}|$ can be any fixed value since it does not enter 
the model for bin counts when $\sin^2 2\theta=0$. In this case, 
the null hypothesis is specified to be $H_0$: $\paratheta=\paratheta_0$ (i.e. the Standard
Model with three light neutrinos). 
Next, for any $\paratheta_1=(\sin^22\theta_1,|\Delta m_1^2|)$ from the parameter
space $\Theta \times M$, 
specify the alternative hypothesis to be $H_1: \beta= \beta_1$, and perform the following procedure:
\begin{enumerate}
\item From the observed data $x$, obtain
\begin{equation*}
  \Delta T(x) := T^{\min}_{H_1}(x) - T^{\min}_{H_0}(x).
\end{equation*} 
\item From the Asimov data set $x^{\text{Asimov}}_{H_0}$, obtain 
\[\Tbarn=\Delta T(x^{\text{Asimov}}_{H_0}) =T^{\min}_{H_1}(x^{\text{Asimov}}_{H_0})\] according to Eq.~\eqref{eq:Tbarn}.
  Then according to the main result that we prove in
  Sec.~\ref{sec:Gaussian}, under $H_0$,
  $\Delta T(X)$ follows approximately a Gaussian
  distribution with mean $\overline{\Delta T_{{H_0}}}$ and standard deviation
  $2\sqrt{|\overline{\Delta T_{{H_0}}}|}$. This suggests that one can approximate $1-p_0$ using
\begin{eqnarray}
  1-p_0 \approx \frac{1+ {\rm Erf}\left(\frac{\Tbarn - \Delta T (x)}{\sqrt{8 |\Tbarn|}} \right)}{2},
\end{eqnarray}
where ${\rm Erf}(s) = \frac{2}{\sqrt{\pi}} \int_0^s e^{-t^2}dt$
is the Gaussian error function for any $s \in (-\infty, \infty)$.
\item Similarly, from the Asimov data set $x^{\text{Asimov}}_{H_1}$, 
obtain
  \[\Tbara =\Delta T(x^{\text{Asimov}}_{H_1})=- T^{\min}_{H_0}(x^{\text{Asimov}}_{H_1})\,.\]
  according to Eq.~\eqref{eq:Tbara}.
Then one can approximate $1-p_1$ using
\begin{eqnarray}
  1-p_1 \approx \frac{ 1+ {\rm Erf}\left(\frac{\Tbara - \Delta T (x)}{\sqrt{8|\Tbara|}} \right)}{2}.
\end{eqnarray}
\item According to Eq.~\eqref{eq:cls}, the CL$_s$ value at
  ($\sin^22\theta_1,|\Delta m_1^2|$) can be approximated by
  \begin{equation} {\rm CL}_s \approx \frac{1+ {\rm Erf}\left(\frac{\Tbara -
          \Delta T (x)}{\sqrt{8 |\Tbara|}} \right)}{1+ {\rm
        Erf}\left(\frac{\Tbarn - \Delta T (x)}{\sqrt{8 |\Tbarn|}}
      \right)}\,.
\end{equation}
The point ($\sin^22\theta_1,|\Delta m_1^2|$) is assigned to the 95\%
CL$_s$ exclusion set if and only if its CL$_s$ value is smaller than
5\%.
\end{enumerate}
In terms of the computing effort, the above CL$_s$ procedure requires the
calculation of $\Delta T (x)$, $\Tbara$, and $\Tbarn$ at each parameter
point in $\Theta \times M$. In comparison, the standard 
Wilks CI method based on $\Delta \chi^2(x)$
in Eq.~\eqref{eq:chisqmin} requires the calculation of $\Delta
\chi^2(x)$ at each parameter point. So the computing cost of
the Gaussian CL$_s$ method is about three times that of the Wilks' CI method. 
In summary,  both methods are easily affordable computationally, but the 
CL$_s$ method is valid under much less restrictive conditions (CD1-CD3).

\section{An Example: Search for Sterile Neutrino }\label{sec:example}

In this section, we introduce an example based on the search for a sterile neutrino.
In this example,  various methods to carry out the CL$_s$ approach and the 
CI approach are compared. 


\subsection{Model Description}\label{sec:toymodel}
In this model, there are two detectors and one neutrino source.
One detector is located at
300 kilo-meters from the neutrino source and is called the near
detector. The other detector is located at 1000 kilo-meters from the neutrino
source and is called the far detector. As shown in
Fig.~\ref{fig:spectrum}, the neutrino energy $E_{\nu}$ covers from 1
GeV to 9 GeV, and a flat (energy independent) neutrino energy spectrum is assumed. 
We further assume the detector can measure the spectrum with 20 energy bins equally 
spaced between 1 GeV and 9 GeV. The
mean number of neutrino events seen by the near (far) detector
without any oscillation is 10 k (0.9 k) per bin. We consider two types
of oscillation measurements: a disappearance measurement with 
oscillation formula 
\begin{equation}\label{eq:disapp}
P_{dis} = 1 - \sin^22\theta \cdot \sin^2 \left( 1.27 \cdot \Delta m^2 \frac{L}{E_{\nu}} \right)\,,
\end{equation}
and an appearance measurement with  oscillation formula 
\begin{equation}\label{eq:app}
P_{app} = \sin^22\theta \cdot \sin^2 \left( 1.27 \cdot \Delta m^2 \frac{L}{E_{\nu}} \right)\,,
\end{equation}
where $\theta$ is the neutrino mixing angle, $\Delta m^2$ is
the neutrino mass squared difference, and $L$ is the distance that 
neutrino travels.

\begin{figure*}
\centering
\includegraphics[width=150mm]{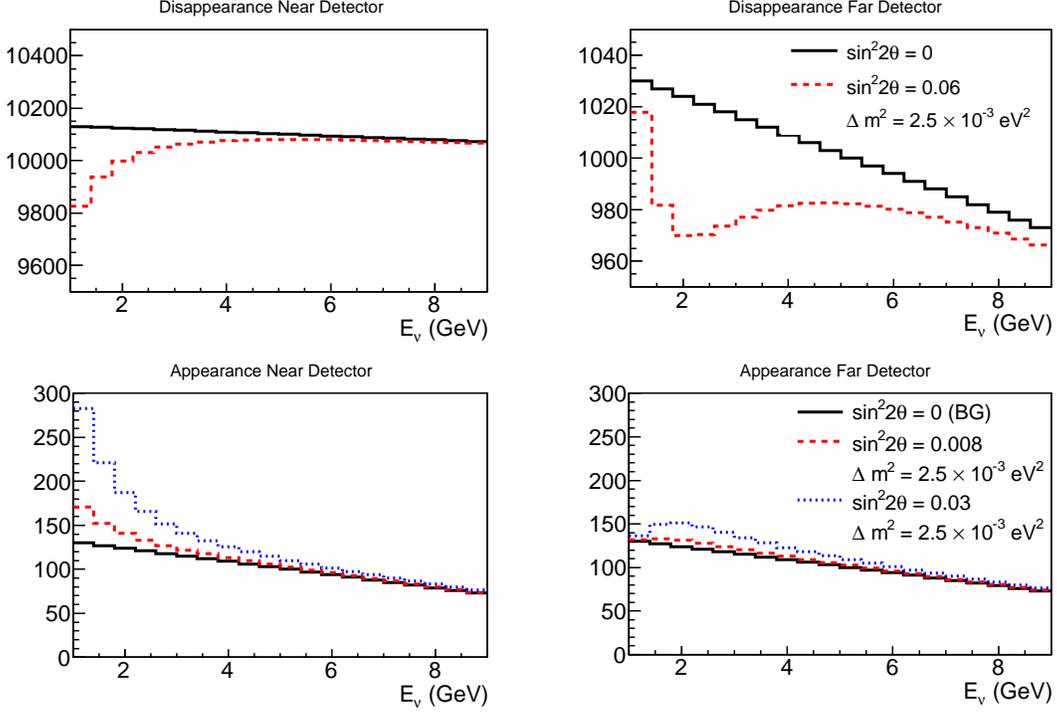}
\caption{(color online) Top panels show the mean number of events seen at the near and 
far detectors in a disappearance experiment. Bottom panels show the mean number of events seen 
at near and far detectors in an appearance experiment. Left and right panels show 
near and far detectors, respectively. See text for more explanations. }
\label{fig:spectrum}
\end{figure*}

We further include a background with a linear dependence on
$E_{\nu}$. The number of background events starts from 130 per bin for
the first bin to 73 per bin for the last (20th) bin. There are three
nuisance parameters, $\epsilon$, $\eta_n$, and $\eta_f$. 
The first one is associated with the detector
efficiency and the neutrino flux, which is assumed to be accurate to
5\%. This uncertainty is assumed to be correlated between the near and
the far detectors. The second and the third nuisance parameters are associated
with the background normalization factors for the near and the far detectors,
respectively. The normalization uncertainty is assumed to be 2\% and
uncorrelated between the two detectors.  Fig.~\ref{fig:spectrum}
shows the expected neutrino spectra. For the disappearance
measurement, we compare the no-oscillation spectrum 
(the null hypothesis $H_0$: $\sin^22\theta = 0$) with an oscillation 
spectrum (an alternative hypothesis $H_1$: $\sin^22\theta = 0.06$ 
at $\Delta m^2 = 2.5\times 10^{-3}$ eV$^2$). For the appearance measurement, we
compare the no-oscillation spectrum (the null hypothesis $H_0$: 
$\sin^22\theta = 0$) with two oscillation spectra (two alternative 
hypotheses $H_1$: $\sin^22\theta = 0.008$ or $\sin^22\theta = 0.03$ at 
$\Delta m^2 = 2.5\times 10^{-3}$ eV$^2$). 
Given a Monte Carlo (MC) sample $N^j_i$, we use the following test statistic based 
on the Poisson likelihood, in line of Eq.~\eqref{eq:poisson_T}: 
\begin{widetext}
\begin{equation}\label{eq:toyT}
  T =  \sum_{j=n,f}\sum_{i=1}^{20} 2\left(\mu^j_i(\epsilon,\eta_j,\sin^22\theta,\Delta m^2) - N^j_i + N^j_i \log \frac{N^j_i}{\mu^j_i(\epsilon,\eta_j,\sin^22\theta,\Delta m^2)} \right) + \frac{\epsilon^2}{0.05^2} +\frac{\eta^2_n}{0.02^2} + \frac{\eta^2_f}{0.02^2}. 
\end{equation}
\end{widetext}
Here, $i$ represents the bin number and ranges from 1 to 20. $j$ labels
the near or the far detector. $\mu_i^j$ is the mean number of
events in $i$-th bin and $j$-th detector.  It depends on the oscillation
parameters: $\sin^22\theta$ and $\Delta m^2$, and the nuisance
parameters: $\epsilon$ for the detector efficiency and neutrino flux, and
$\eta_n$ ($\eta_f$) for the near (far) detector background
normalization factors.

\subsection{The Wilks' CI method vs. the MC CI method}\label{sec:FCvsW}
For the example mentioned above, the Wilks' method is unsuitable for setting CI for the parameter
$\sin^22\theta$ because the conditions required are not satisfied as stated in Sec.~\ref{sec:review}.
In comparison, the computationally intensive MC CI method need to be used 
to set CI in this example. The purpose of this section is to demonstrate the practical difference 
between the two methods. Here, we examine the distribution of the test statistic in 
Eq.~\eqref{eq:toychi} under the hypothesis $H_0$: $\sin^22\theta = 0$, where
the Wilks' method is especially problematic. To implement the MC CI method, 
we generate a large number of MC samples assuming that $\sin^22\theta = 0$. The MC 
samples have statistical fluctuations according to Poisson distributions, and systematic
variations through randomizing the three nuisance parameters according to normal
distributions. While the minimization process in calculating T$^{\min}$
follows the Frequentist's approach, the randomization of the nuisance parameters 
corresponds to a Bayesian integral over the nuisance parameters. It is a common 
hybrid Bayesian/Frequentist approach~\cite{Cousins:1991qz}. As a comparison, we also 
tried a full Frequentist approach as illustrated in Ref.~\cite{full_freq,hybrid_sample}. Results 
are very similar to that of the hybrid approach. In the latter
approach, MCs are generated using the best-fit nuisance parameters obtained in 
analyzing the data under the $\sin^22\theta = 0$ hypothesis.

For each MC sample, we find $T^{\min}$ and $T_{H_0}^{\min}$, where $T^{\min}$ is
the minimum value of $T$ from Eq.~\eqref{eq:toyT} in the 5-dimensional parameter
space of ($\sin^22\theta$, $\Delta m^2$, $\epsilon$, $\eta_n$,
$\eta_f$), and $T_{H_0}^{\min}$ is the minimum value of $T$ under the restriction, 
$\sin^22\theta_{\text{true}} = 0$. 
Then we form the test statistic
\begin{equation} \label{eq:toychi}
\Delta \chi^2 = T_{H_0}^{\min} - T^{\min}\,.
\end{equation}
Fig.~\ref{fig:notchisq} shows the distribution of $\Delta
\chi^2$, which clearly does not follow a Chi-square distribution
with two degrees of freedom. In summary, for this example, the Wilks' method can not be 
used to correctly set CIs based on the test statistic $\Delta \chi^2$. It is possible to 
explore alternative formula than that of the Wilks' method, if one takes the hybrid 
approach in Ref.~\cite{Cousins:1991qz} and finds an analytic approximation to the 
solution of $t_c$ for the equation Prob$(\Delta \chi^2 \leq t_c)\geq c$, where the 
probability is evaluated over the distribution of the nuisance parameters. Otherwise, 
one can always obtain the distribution of $\Delta \chi^2$
through the computationally intensive MC CI method.

\begin{figure}
\centering
\includegraphics[width=75mm]{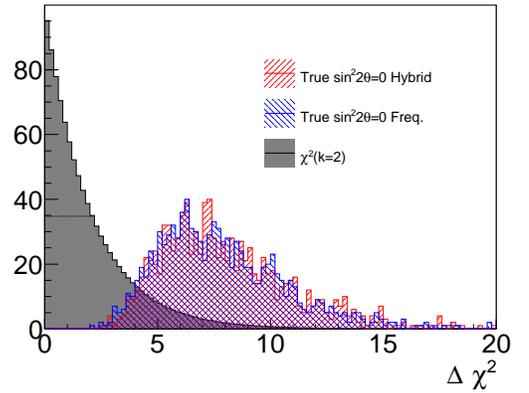}
\caption{(color online) Distributions of $\Delta \chi^2 =
  T(\sin^22\theta = 0) - T_{\min}$ is plotted for MCs with the true
  $\sin^22\theta = 0$.  Distributions based on the hybrid Bayesian/Frequentist 
and full Frequentist approaches are compared to the Chi-square
  distribution with two degrees of freedom.}
\label{fig:notchisq}
\end{figure}

\subsection{Validity of the Gaussian Approximation in the Gaussian 
CL$_s$ method}\label{sec:GCLvsMC}

For the example in the previous section, there is no known
way to set CI without computationally intensive MC simulations. 
This is the main motivation for using the Gaussian CL$_s$ method as an alternative. 
In this section, 
we check how closely does the CL$_s$ test statistic 
$\Delta T = T_{{H_1}}^{\min} -
T_{{H_0}}^{\min}$ follow the normal distribution N($\overline{\Delta
  T}$,4$\overline{\Delta T}$).  Here, $T_{H}^{\min}$ is the value of the test 
statistic $T$ from Eq.~\eqref{eq:toyT} under the hypothesis $H$, 
minimized over the nuisance parameters $(\epsilon,\eta_n,\eta_f)$.
 Fig.~\ref{fig:dis_com} shows
the distribution of $\Delta T$ for the disappearance measurement. The
null hypothesis $H_0$ corresponds to $\sin^22\theta = 0$.  The
alternative hypothesis $H_1$ corresponds to $\sin^22\theta = 0.06$ and
$\Delta m^2 = 2.5\times 10^{-3}$ eV$^2$. The histograms on the left
(right) are made from the MC samples assuming $H_1$ ($H_0$) is
true. We also compare them with the expected normal distribution
N($\overline{\Delta T}$,4$\overline{\Delta T}$) from the
$\overline{\Delta T_{{H_0}}}$ and $\overline{\Delta T_{{H_1}}}$ calculated
from the Asimov data sets. Good agreements are observed.

Similarly, we also check the appearance measurements.  In
Fig.~\ref{fig:ap1_com}, the null hypothesis $H_0$ corresponds to
$\sin^22\theta = 0$,  and the alternative hypothesis $H_1$ corresponds to
$(\sin^22\theta, \Delta m^2) = (0.008, 2.5\times 10^{-3}$eV$^2)$.
In Fig.~\ref{fig:ap2_com}, $H_0$ corresponds to
$\sin^22\theta = 0$, and $H_1$ corresponds to
$(\sin^22\theta, \Delta m^2) = (0.03, 2.5\times 10^{-3}$eV$^2)$.
The agreement between the MCs and expectations in
Fig.~\ref{fig:ap1_com} is slightly worse than that in
Fig.~\ref{fig:dis_com}, but is still reasonably good.  However, the
difference between the MCs and expectations in Fig.~\ref{fig:ap2_com}
becomes large. 
This is because the third regularity condition CD3 ``when the prediction of
two hypotheses (the null hypotheses $H_0$ and the alternative
hypothesis $H_1$ are relatively close or $|\mu_i-\nu_i| << \mu_i \sim
\nu_i$'' is no longer met. 
In a disappearance search, CD3 can be easily
satisfied.  However, this may not be true in an appearance experiment
as the mean number of signal events is zero when $\sin^22\theta = 0$. 
When $H_0$ and  $H_1$ are 
$\sin^22\theta = 0$ and $(\sin^22\theta, \Delta m^2) = (0.008, 2.5\times 10^{-3}$eV$^2)$, respectively, 
CD3 is still reasonably well satisfied with the existence of
backgrounds.  When $H_0$ and $H_1$ are 
$\sin^22\theta = 0$ and $(\sin^22\theta, \Delta m^2) =
(0.03, 2.5\times 10^{-3}$eV$^2)$, respectively, CD3 is 
severely violated. Note that in such situations where $H_0$ and $H_1$ 
are very different, the experimental data is most likely able to exclude 
one hypothesis easily, making it less interesting to carry out such a
 statistical test.
{\bf Nevertheless, we emphasize it is crucial to validate Gaussian approximation with MCs in practice
when any of CD1-CD3 listed in the end of Sec.~\ref{sec:DeltaT} 
are marginally satisfied.}

\begin{figure}
\centering
\includegraphics[width=75mm]{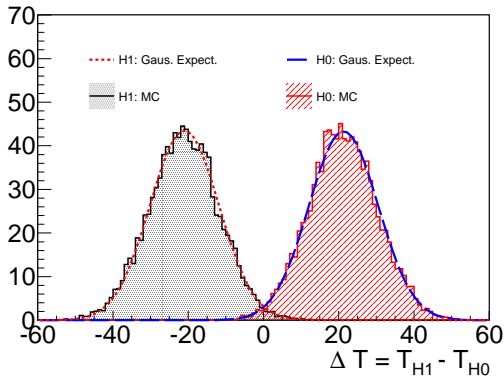}
\caption{(color online) The test statistic $\Delta T = T_{H_1}^{\min} - T_{H_0}^{\min}$ is plotted
for MCs (disappearance) assuming the hypothesis $H_0$ or $H_1$ is true. 
Here, the null hypothesis $H_0$ corresponds to $\sin^22\theta = 0$. The alternative hypothesis 
$H_1$ corresponds to $\sin^22\theta = 0.06$ and $\Delta m^2 = 2.5\times 10^{-3}$ eV$^2$.
}
\label{fig:dis_com}
\end{figure}

\begin{figure}
\centering
\includegraphics[width=75mm]{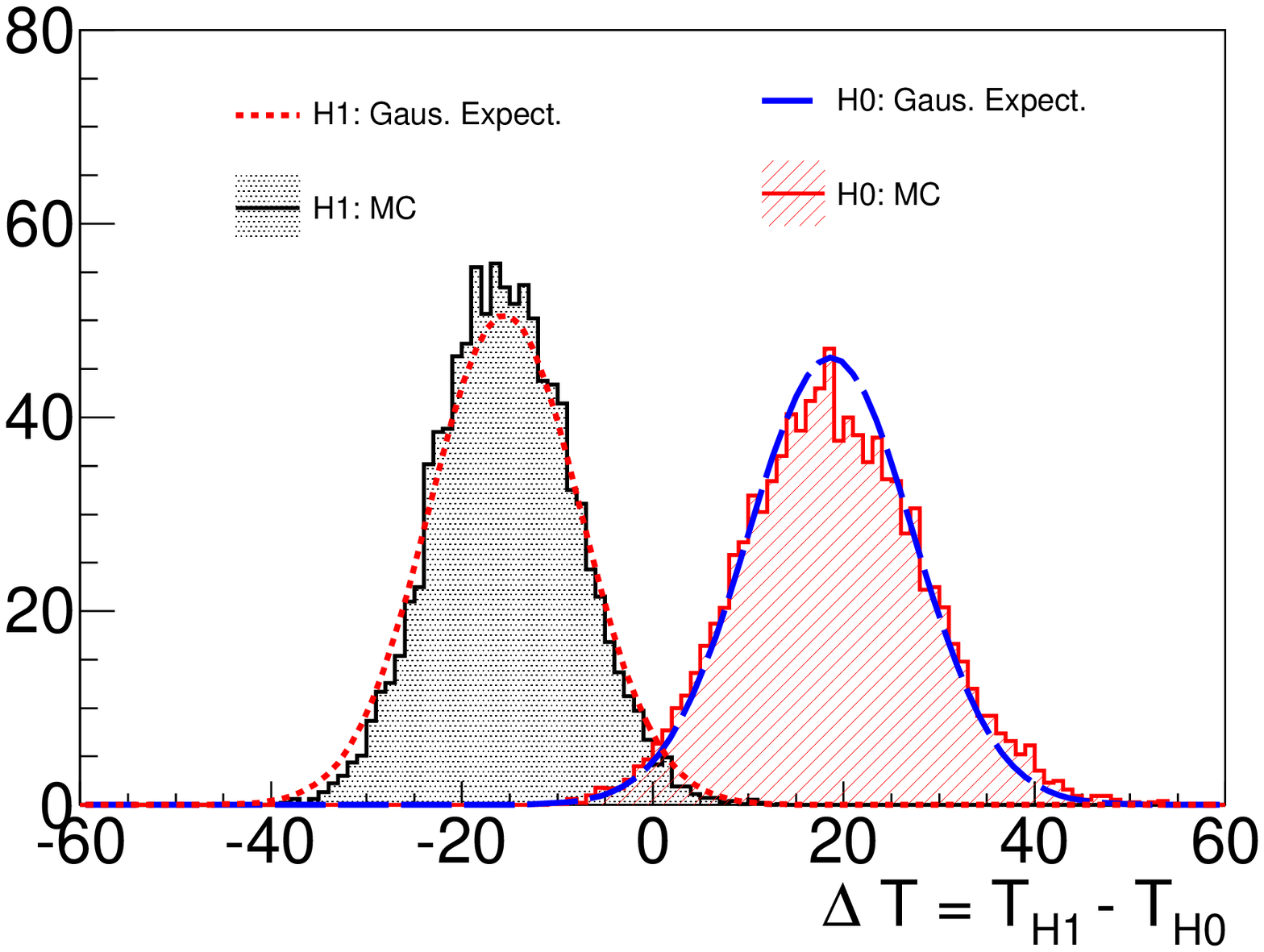}
\caption{(color online) The test statistic $\Delta T = T_{{H_1}}^{\min} - T_{{H_0}}^{\min}$
 is plotted
for MCs (appearance) assuming the hypothesis $H_0$ or $H_1$ is true. 
Here, the null hypothesis $H_0$ corresponds to $\sin^22\theta = 0$. The alternative hypothesis 
$H_1$ corresponds to $\sin^22\theta = 0.008$ and $\Delta m^2 = 2.5\times 10^{-3}$ eV$^2$.}
\label{fig:ap1_com}
\end{figure}

\begin{figure}
\centering
\includegraphics[width=75mm]{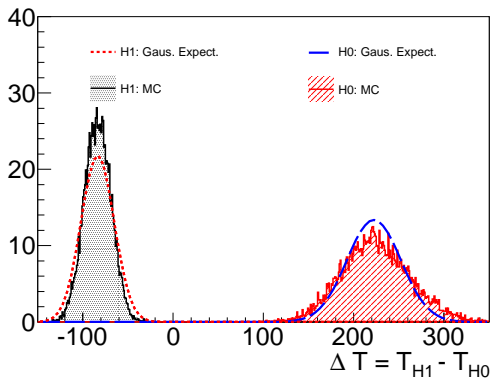}
\caption{(color online) The test statistic $\Delta T = T_{{H_1}}^{\min} - T_{{H_0}}^{\min}$ is plotted
for MCs (appearance) assuming the hypothesis $H_0$ or $H_1$ is true. 
Here, the null hypothesis $H_0$ corresponds to $\sin^22\theta = 0$. The alternative hypothesis 
$H_1$ corresponds to $\sin^22\theta = 0.03$ and $\Delta m^2 = 2.5\times 10^{-3}$ eV$^2$.}
\label{fig:ap2_com}
\end{figure}

\subsection{The Gaussian CL$_s$ method vs. the $\Delta T$-based CI method, both based on $\Delta T = T_{{H_1}}^{\min} - T_{{H_0}}^{\min}$} \label{sec:CLvsCI}

\begin{figure}
\centering
\includegraphics[width=75mm]{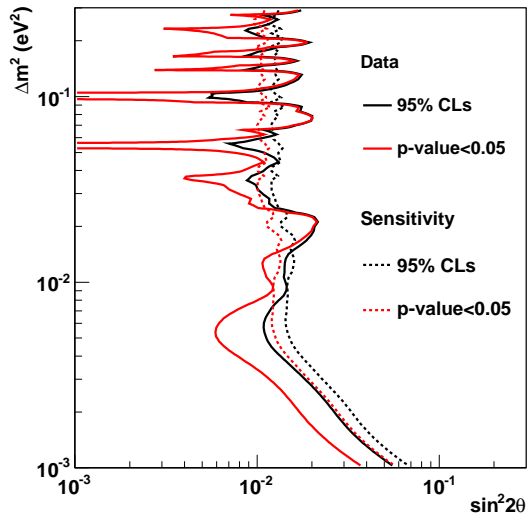}
\caption{(color online) 
Comparison of the exclusion sets determined by the Gaussian CL$_s$ method vs. the
CI, both using the 
test statistic $\Delta T = T_{{H_1}}^{\min} - T_{{H_0}}^{\min}$. The true value of 
$\sin^22\theta$ is $0$. For the CI (CL$_s$) method, the right 
side of the red (black) line has a p-value (CL$_s$ value) smaller than 
$0.05$. The sensitivity curves are generated from a large number of Monte Carlo 
samples. At each $\Delta m^2$, 50\% (50\%) of 
MC samples will have a better (worse) exclusion limit than the sensitivity curve.}
\label{fig:pvvCLs}
\end{figure}

Recall that the CL$_s$ approach is based on the test statistic 
$\Delta T = T_{{H_1}}^{\min} - T_{{H_0}}^{\min}$. Although typical methods to form CIs use a 
different type of test statistic, namely $\Delta \chi ^2$ shown in the previous sections, 
one can in principle also set CIs based on $\Delta T$, which we refer to as the $\Delta T$-based CI method. 
Below, we use our example to compare the exclusion sets obtained from the Gaussian CL$_s$ method and the CIs obtained from the $\Delta T$-based CI method. For the former method, the exclusion set consists of parameter values that correspond to CL$_s$ values, specifically $(1-p_1)/(1-p_0)$, lower than 0.05; and for the latter method, the CI consists of parameter values that correspond to p-values, specifically $(1-p_1)$, over 0.05. The results are summarized in Fig.~\ref{fig:pvvCLs}. 
In the example, the true $\sin^22\theta$ is assumed to be zero. 
The sensitivity of the Gaussian CL$_s$ method is slightly worse 
than that of $\Delta T$-based CI method, because the CL$_s$ value is by construction 
larger than the p-value used in the CI method. 
Despite the slightly worse sensitivity, the CL$_s$ produces smoother contours that 
agree better with intuition (not excluding hypotheses that are close to the null hypothesis) 
than the CI do.
As shown in Fig.~\ref{fig:pvvCLs}, the region 
that correspond to $\Delta m^2 \sim 5.5 \times10^{-2}$ eV$^2$ and 
$\sin^22\theta < 0.01$ (also $\Delta m^2 \sim 0.1$ eV$^2$ and 
$\sin^22\theta < 0.01$) is excluded from the 95\% CI. This is inconsistent with intuition as the 
expected spectrum for small $\sin^22\theta$ values should be very similar 
to that of $\sin^22\theta = 0$, and we do not expect to exclude regions 
with small $\sin^22\theta$ values. This phenomenon can be understood 
as follows. With the test statistic $\Delta T$, we compare 
two hypotheses each time. Therefore, even when 
the two hypotheses are very similar, the chance of excluding one hypothesis with 
CI can still be large as illustrated in Fig.~\ref{fig:CLs2}. 
As we explained in Sec.~\ref{sec:CLs_1}, the definition of the CL$_s$ value
avoid this problem, giving it an advantage over the traditional 
CI when test statistic $\Delta T$ is used.

\subsection{The Gaussian CL$_s$ method vs. the MC CI method vs.  the Raster-Scan MC CI 
method}

The statistical interpretation of (the complement of) exclusion sets obtained using the CL$_s$ method 
is distinct from that of CIs. Indeed, if an exclusion contour based on thresholding the CL$_s$ value at $\alpha$ is used to set a CI, its coverage probability will be over $1-\alpha$.   
Nevertheless, it is still interesting to compare these two kinds of sets in 
specific physics problems, as seen in many literatures 
(for example, Ref.~\cite{Read:2002hq} and Ref.~\cite{cs11}). Below, we perform such a
comparison under the set up of our example. Besides the CL$_s$ approach and 
the standard $\Delta \chi^2$-based CI approach, we also include results from another commonly used 
approach, the so-called raster-scan CI approach. In short, this approach scans 
through all values of the parameter $|\Delta m^2|$, and at each fixed $|\Delta m^2|$, 
it checks the compatibility of the other parameter $\sin^22\theta$ to the data. 
A most popular method to carry out the raster scan approach uses the 
following statistic at each $|\Delta m^2|$,
\begin{widetext}
\begin{equation}
\Delta \chi^2_{RS}(\sin^22\theta,|\Delta m^2|;x) \equiv 
\chi^2(\sin^22\theta,|\Delta m^2|,\eta_{\min}(\sin^22\theta,|\Delta m^2|); x) - \chi^2_{RS\min}(|\Delta m^2|;x)\,,
\end{equation}
\end{widetext}
which is similar to the $\Delta \chi^2$ statistic given in Eq.~\eqref{eq:chisqmin} 
except that the global minimum $\chi^2_{\min}(x)$ is replaced by the restricted 
minimum $\chi^2_{RS\min}(|\Delta m^2|;x)=\min_{\sin^22\theta,\eta}\chi^2(\sin^22\theta, |\Delta m^2|,\eta; x)$. 
Given a fixed value of $|\Delta m^2|$, the raster scan method examines all 
$\sin^22\theta_1$ values, one at a time, and test the hypothesis $H_0$: 
$\sin^22\theta=\sin^22\theta_1$ 
based on the statistic $\Delta \chi^2_{RS}(\sin^22\theta_1,|\Delta m^2|;x)$. 
The raster scan approach is usually considered less ideal than the standard CI approach 
that we described in Sec.~\ref{sec:review}, mainly because it does not make comparisons 
between hypotheses that have different values of $|\Delta m^2|$ and hence can not 
distinguish a likely value of this parameter from an unlikely one~\citep{Feldman:1997qc}. 
In addition, according to Eq.~\eqref{eq:disapp} 
and Eq.~\eqref{eq:app}, when $\sin^22\theta=0$, any value of $|\Delta m^2|$ will 
result in the same model, namely, the Standard Model. As a consequence, 
the Standard Model is tested many times against different new physics hypotheses 
that correspond to different values of $|\Delta m^2|$, which makes it difficult to 
interpret the test results. Whereas in the standard CI approach, any model is 
tested only once.  
Similar to the case of $\Delta \chi^2$, the regularity condition of the Wilks' theorem
would also break for $\Delta \chi^2_{RS}$ when the true $\sin^22\theta=0$. Therefore,
Monte Carlo is usually necessary to obtain the distribution of $\Delta \chi^2_{RS}$ 
to compute CIs using the raster scan.  

Fig.~\ref{fig:FCvCLs_sen} compares the sensitivity of the Gaussian CL$_s$ method, 
the standard MC CI method, and the raster-scan MC CI method. 
We assumed that the true value of $\sin^22\theta$ is $0$ in generating
MC. At each $\Delta m^2$, 50\% (50\%) of 
MC samples will have a better (worse) exclusion limit than the sensitivity curve.
Sensitivities from these three methods are similar. 
The sensitivity of the 95\% exclusion set from the Gaussian CL$_s$ method is 
slightly better than that of the 95\% CI from the MC CI method, and is in fact
close to that of the 90\% CI from the MC CI method for this setup. This is expected, 
since the test statistic $\Delta T = T_{{H_1}}^{\min} - T_{{H_0}}^{\min}$  used in the 
Gaussian CL$_s$ method is designed to focus on the differences between the new 
physics hypotheses ($H_1: \sin^22\theta =\sin^22\theta_1$ for some 
$\sin^22\theta_1  > 0$), with the Standard Model ($H_0: \sin^22\theta = 0$). 
Therefore, when the true value of $\sin^22\theta$ is $0$, the Gaussian CL$_s$ method has 
larger power to exclude new physics hypotheses than the MC CI method. 
In addition, the 95\% sensitivity from the Gaussian CL$_s$
method is very close to that from the raster-scan MC CI method. This is actually 
a coincidence, since the CL$_s$ method and the raster-scan method use 
the ratios of p-values and p-value to set limits, respectively. The left panel 
of Fig.~\ref{fig:FCvCLs_com} shows the difference between the CL$_s$ sensitivity 
(CL$_s$ value) and raster-scan sensitivity (p-value) at each parameter point. 
The difference is rather large at small values of $\sin^22\theta$, which indicates that
the similarity of the 95\% lines between the CL$_s$ method and the raster-scan method 
is a coincidence. The right panel of Fig.~\ref{fig:FCvCLs_com} shows the sensitivity
difference between the raster-scan MC CI method and the standard MC CI method. The 
sensitivity are also different, since the choice of test statistics are different 
between the raster-scan MC CI method and the standard MC CI method.

\begin{figure}
\centering
\includegraphics[width=75mm]{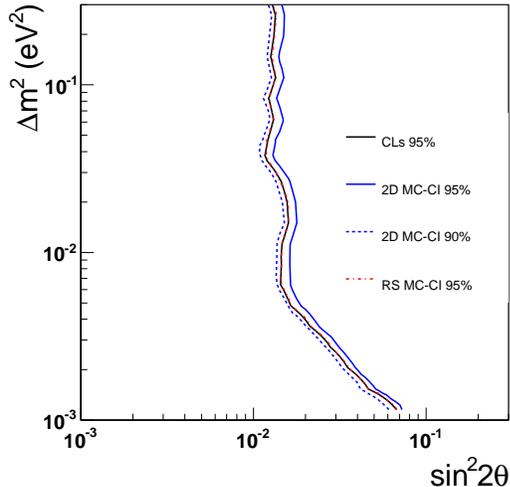}
\caption{(color online) 
Comparison of the sensitivity of the 95\% Gaussian CL$_s$ method vs. 
that of the 95\% and the 90\% MC CI method. We also added the 95\% 
raster-scan MC CI for comparison. The true value of $\sin^22\theta$ is $0$. 
See texts for more explanations.}
\label{fig:FCvCLs_sen}
\end{figure}

\begin{figure*}
\centering
\includegraphics[width=150mm]{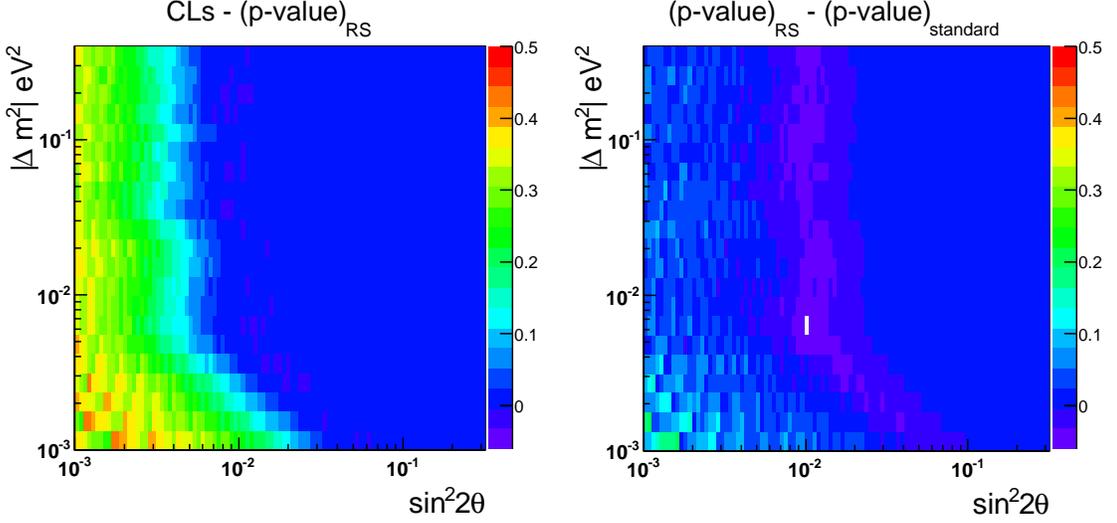}
\caption{(color online) Left panel: the difference between the CL$_s$ sensitivity
and the raster-scan sensitivity is shown at each parameter point. Right panel:
the difference between the raster-scan sensitivity and the standard CI sensitivity
is shown at each parameter point. }
\label{fig:FCvCLs_com}
\end{figure*}

When the new physics is indeed true, the standard MC CI method has  
clear advantage in constraining the parameter space over the other two methods. This is shown in
Fig.~\ref{fig:FCvCLs}. The MC sample is generated with
$\sin^22\theta_{\text{true}} = 0.1$ and $\Delta m^2_{\text{true}} =
2.5\times10^{-3}$ eV$^2$ with statistical fluctuations and systematic
variations. The 90\% CI of the MC CI method were able to identify a small region
close to the true value. In comparison, 
the 95\% CL$_s$ limit successfully excluded the region on the right, but failed 
to exclude regions (on the left of line) far away from the true value.
This again is due to the choice of the test statistic ($\Delta T$ in the Gaussian 
CL$_s$ method vs. $\Delta \chi^2$ in the MC CI method).  The proposed
test statistic $\Delta T$ focuses on the difference between the 
new physics hypothesis and the Standard Model, while the
test statistic $\Delta \chi^2$ takes into account all the 
likely values of $(\sin^22\theta, \Delta m^2)$. Therefore, we confirm the
conclusion from Ref.~\cite{Read:2002hq}: ``the CL$_s$ technique for
setting limits is appropriate for determining exclusion sets
while the determination of CIs advocated by the Feldman-Cousins method is more 
appropriate for treating established signals''. For comparison, we also display 
the 95\% raster-scan MC CI. Since the raster scan can not distinguish likely and unlikely values of
the parameter $|\Delta m^2|$, it also failed to exclude some regions of the parameter
space that are far away from the truth.

\begin{figure}
\centering
\includegraphics[width=75mm]{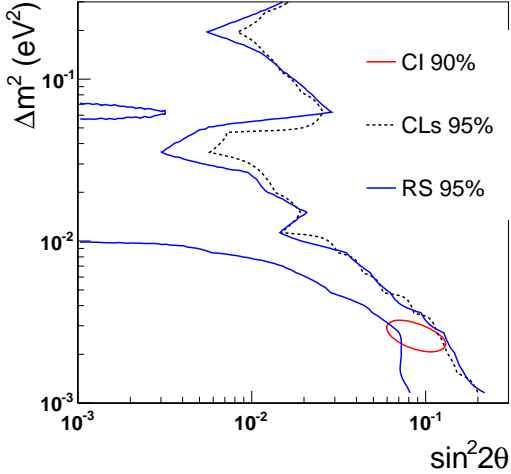}
\caption{(color online) Comparison of the 90\% CIs 
from the MC method vs. 95\% Gaussian CL$_s$ limit
for a MC sample generated with $\sin^22\theta = 0.1$ and 
$\Delta m^2 = 2.5\times10^{-3}$ eV$^2$. The 95\% raster-scan MC CI
is shown for comparison. See texts for more explanations.}
\label{fig:FCvCLs}
\end{figure}

\section{Discussion}\label{sec:discussion}


%

In order to use the Gaussian CL$_s$ method, it is important that the CD1-CD3 listed 
in the end of Sec.~\ref{sec:DeltaT} are met. The first condition CD1 is continuity of the 
parameter space for the nuisance parameters, 
under both the null and the alternative hypotheses. This requirement is easier to achieve
compared to the first regularity condition required by the Wilks' theorem,
since it concerns the nuisance parameters only, not the parameters of interest, 
($\sin^22\theta,|\Delta m^2|$). The second condition CD2 concerns large enough data size, 
which is also easier to reach compared to that required by the 
Wilks' theorem. This is because in the Gaussian CL$_s$ method tests a simpler pair of hypotheses, in which the values of ($\sin^22\theta,|\Delta m^2|$) 
are fixed, and one automatically avoids the situation shown in Fig.~\ref{fig:illustration}b 
that involves minimization over a large range of $|\Delta m^2|$ in calculating the 
test statistic. The third condition CD3 is that the difference between the predictions
of two hypotheses is small comparing to the predictions themselves. 
In searching for new physics with precision measurements, the signal from
the Standard Model is usually much larger than the potential signal from new 
physics. Therefore, CD3 is generally satisfied.
In the case that CD3 is violated or marginally satisfied 
(see Fig.~\ref{fig:ap2_com}), one should use Monte Carlo simulation to 
derive the distribution of the test statistic. 


Similar to the Wilks' 
CI method based on the test statistic $\Delta \chi^2$
and predefined constants, the Gaussian CL$_s$ method also allows easy combination 
of multiple independent experimental results that probe the same parameter space. 
The CL$_s$ value at each alternative hypothesis 
$H_1$ from experiments ($k=1,2,...m$) can be calculated with
\begin{eqnarray}
\Delta T 
(x) &=& \sum^{m}_{k=1} \Delta T(x_k)
, \nonumber \\
\overline{\Delta T(x^{\text{Asimov}}_{H_1})} &=& \sum^{m}_{k=1} \overline{\Delta T(x^{Asimov~k}_{H_1})},\nonumber \\
\overline{\Delta T(x^{\text{Asimov}}_{H_0})} &=& \sum^{m}_{k=1} \overline{\Delta T(x^{Asimov~k}_{H_0})}, \nonumber
\end{eqnarray}
with $x = \sum^{m}_{k=1}$ is the combined data from all experiments. 
This can be easily achieved if each experiment releases their maps of
$T_{x_k}$,
$\overline{\Delta T(x^{\text{Asimov}}_{H_1})}$, and $\overline{\Delta T(x^{\text{Asimov}}_{H_0})}$. In practice, the main challenge
in combining multiple experiment results arise from the potential 
correlation among different 
experiments and requires careful examinations.

So far, we have argued that, in practice, the CL$_s$ method is often simple to use and allows easy combination of multiple results. But it is important to remind the readers that the CL$_s$ is a limited method that aims at setting boundaries only. The CL$_s$ based on $\Delta T$ does not directly address the question ``do we see new physics or not'', nor does it provide estimate of parameters. To help address the first question, we recommend reporting the p-value based on 
the test statistic $\Delta \chi^2$ assuming the Standard Model is true, in addition to the obtained exclusion sets. To address the second question, the standard CI approach is needed. Indeed, the standard CI approach is the preferred approach to take whenever one can afford to carry it out correctly, because the standard CI approach is a unified approach to set limits in the absence of new physics signals and to estimate parameters after the discovery of new physics \cite{Feldman:1997qc}.

\section{Summary}\label{sec:summary}

In this paper, we describe a method to present results in searching 
for new physics in a continuous parameter space. 
This method takes the CL$_s$ approach to obtain exclusion sets for parameters.
Specifically, the method consists of testing many pairs of hypotheses. Each time, 
a new physics model is tested against the Standard Model using the log-likelihood ratio test statistic, 
or certain variations 
of it, denoted by $\Delta T$. We provide a mathematical proof 
to show that the distribution of $\Delta T$ follows a Gaussian 
distribution at large data limit under either hypothesis, when the two hypotheses are relatively 
close. This result allows a simple alternative to the computationally intensive Monte Carlo method to calculate CL$_s$ values,
and thus to set exclusion limits in one or multiple dimensional parameter spaces. This method can also be 
used to conveniently combine results from multiple experiments.

\section{Acknowledgments}
We would like to thank Wei Wang for helpful discussions. 
This material is based upon work supported by the National Science 
Foundation and the U.S. Department of Energy, Office of Science, 
Office of High Energy Physics, Early Career Research program 
under contract number DE-SC0012704.

\section*{Appendices}

\subsection{A few basic properties of the fitted models under $H_0$ and $H_1$}\label{app:hats}
Suppose $H_0$ is the correct hypothesis, that is, the data $X$ came from $H_0$. Having observed the data, the best fitting models under $H_0$ and $H_1$ have estimated nuisance parameters $\hnutn$ and $\hnuta$ respectively, as defined in Sec.~\ref{sec:CLs}. The corresponding per unit mean counts are denoted $\hpi$ and $\htau$ respectively. 

We show below that there is a unique limit of $\hnuta$ as the data size increases, and that it leads to the model $\nu(\nuta)$ that is the closest model under $H_1$
 to the true model $\muT$ under a certain criteria.
 Indeed, let $t_m(\nuta)=\chi^2_{\IH}(\nuta)/m=\sum_i\frac{\left(\p_i-\tau_i(\nuta)\right)^2}{\tau_i(\nuta)}$, 
and let $t(\nuta)=\sum_i\frac{\left(\pi^0_i-\tau_i(\nuta)\right)^2}{\tau_i(\nuta)}$. Since $\p$ converges almost surely (a.s.) to $\piT$ as $m$ increases, we have $t_m(\nuta)$ converges a.s. to $t_m(\nuta)$. Then under regularity conditions, such as $t_m$ being twice differentiable and convex in $\nuta$, 
$\hnuta^m=\arg\min_\nuta t_m(\nuta)$ also converges a.s. to $\arg\min_\nuta t(\nuta)$  as $m$ increases. By denoting the limit of $\hnuta^m$ by $\nuta^0$, we have $m t(\nuta^0)=\min_\nuta m t(\nuta)$. That is, $\nuta^0$ and $\nu^0:=\nu(\nuta^0)$ are such that $\sum_{i=1}^\n\frac{(\mu^0_i-\nu^0_i)^2}{\nu^0_i}=\min_{\nu}\sum_{i=1}^\n\frac{(\mu^0_i-\nu_i)^2}{\nu_i}\,$.

We list a few more properties that are useful in the proof of Lemma~1 and the proof of the result in Eq.~\eqref{eq:proof}. It is well-known that $\hnutn-\nutnT$ and $\hpi-\piT$ are both of order $O_p(m^{-\frac{1}{2}})$. 
And $\hnuta-\nutaT$ and $\htau-\tauT$ are also both of order $O_p(m^{-\frac{1}{2}})$ according to Ref.~\cite{whit:1982}.

\begin{widetext}

\subsection{Lemma~\ref{lemma:htau}}
 \begin{lemma}\label{lemma:htau}
Assuming [A0], 
we have
\begin{equation} \sqrt{m}  \begin{pmatrix}
    \p - \pi^0\\[1mm]
    \htau-\tau^0
  \end{pmatrix} 
  = 
   \begin{pmatrix}
    I\\[1mm]
   E^*
  \end{pmatrix} 
   \sqrt{m}\left(\p-\pi^0\right)
   +O_p(m^{-\frac{1}{2}})
   \,,
   \end{equation}   where  \[E^*=B^*\left({B^*}^T \dg\left\{\frac{(\pi^0)^2}{(\tau^0)^3}\right\} B^*\right)^{-1}{B^*}^T \dg\left\{\frac{\pi^0}{(\tau^0)^2}\right\}\;\; \text{and}\;\; B^*_{\n\times q^*}=\frac{\partial \tau^0}{\partial \nuta}\,.\] Further,
   \begin{equation}\label{eq:lemma2}{B^*}^T \dg\left\{\frac{(\pi^0+\tau^0)}{(\tau^0)^2}\right\}\delta=0\,.    \end{equation}
 \end{lemma}
\begin{proof}
By definition, $\hnuta$ is such that \[\frac{\partial \chi^2_{\IH}(\hnuta)}{\partial  \nuta_j}=0\;\;\text{for $j=1,\cdots,q^*.$}\]
That is,
\[2\sum_{i=1}^\n  \frac{{\Ni}+\hnu_i}{\hnu_i^2}\frac{\partial \hnu_i}{\partial\nuta_j}({\Ni}-\hnu_i)=0 \]
\[\sum_{i=1}^\n  \frac{{\Ni}+\hnu_i}{\hnu_i^2}\frac{\partial \hnu_i}{\partial\nuta_j}({\Ni}-\nu^0_i)
=\sum_{i=1}^\n  \frac{{\Ni}+\hnu_i}{\hnu_i^2}\frac{\partial \hnu_i}{\partial\nuta_j}(\hnu_i-\nu^0_i)\]
\begin{equation}\label{eq:lrhs}
\sum_{i=1}^\n  \frac{\p_i+\htau_i}{\htau_i^2}\frac{\partial \htau_i}{\partial\nuta_j}(\p_i-\pi^0_i+\pi^0_i-\tau^0_i)
=\sum_{i=1}^\n  \frac{\p_i+\htau_i}{\htau_i^2}\frac{\partial \htau_i}{\nuta_j}(\htau_i- \tau^0_i)\,.
\end{equation}
Note by delta's method
\[ \htau_i= \tau^0_i +\sum_k \frac{\partial \tau^0_i}{\partial\nuta_k}(\hnuta_k-\nuta^0_k)+ O_p(m^{-1})\;\;\text{and}\;\;\;\frac{\partial \htau_i}{\nuta_j}=\frac{\partial \tau^0_i}{\nuta_j}+\sum_k \frac{\partial^2 \tau^0_i}{\partial\nuta_j \partial\nuta_k}(\hnuta_k-\nuta^0_k)+O_p(m^{-1})\,, \]
and
\[\frac{p_i+\htau_i}{(\htau_i)^2}-\frac{\pi^0_i+\tau^0_i}{(\tau^0_i)^2}
=\frac{1}{(\tau^0_i)^2}(p_i-\pi^0_i)-\frac{1}{(\tau^0_i)^2}(2\frac{\pi^0_i}{\tau^0_i}+1)(\htau_i-\tau^0_i)+O_p(m^{-1})\,.\]
Hence, the lhs of Eq.~\eqref{eq:lrhs} becomes
\[
\begin{split}
 lhs&=\sum_{i=1}^\n   \left(\frac{\pi^0_i+\tau^0_i}{(\tau^0_i)^2}+\frac{1}{(\tau^0_i)^2}(p_i-\pi^0_i)-\frac{1}{(\tau^0_i)^2}(2\frac{\pi^0_i}{\tau^0_i}+1)(\htau_i-\tau^0_i)+O_p(m^{-1})\right)(\frac{\partial \tau^0_i}{\partial\nuta_j})(\p_i-\pi^0_i+\delta_i)\\
&=\sum_{i=1}^\n   \left(\frac{\pi^0_i+\tau^0_i}{(\tau^0_i)^2}\right)\frac{\partial \tau^0_i}{\partial\nuta_j}\delta_i+\sum_{i=1}^\n  \left[ \frac{\pi^0_i+\tau^0_i}{(\tau^0_i)^2}+\frac{
\delta_i}{(\tau^0_i)^2}\right]\frac{\partial \tau^0_i}{\partial\nuta_j}(\p_i-\pi^0_i)-\sum_{i=1}^\n \frac{1}{(\tau^0_i)^2}(2\frac{\pi^0_i}{\tau^0_i}+1)(\htau_i-\tau^0_i)(\frac{\partial \tau^0_i}{\partial\nuta_j})\delta_i+O_p(m^{-1})\\
&=\sum_{i=1}^\n   \left(\frac{\pi^0_i+\tau^0_i}{(\tau^0_i)^2}\right)\frac{\partial \tau^0_i}{\partial\nuta_j}\delta_i+\sum_{i=1}^\n  \frac{2\pi^0_i}{(\tau^0_i)^2}\frac{\partial \tau^0_i}{\partial\nuta_j}(\p_i-\pi^0_i)-\sum_{i=1}^\n \frac{2\pi^0_i+\tau^0_i}{(\tau^0_i)^3}\delta_i (\frac{\partial \tau^0_i}{\partial\nuta_j})(\htau_i-\tau^0_i)+O_p(m^{-1})\\
&=\sum_{i=1}^\n   \left(\frac{\pi^0_i+\tau^0_i}{(\tau^0_i)^2}\right)\frac{\partial \tau^0_i}{\partial\nuta_j}\delta_i+\sum_{i=1}^\n  \frac{2\pi^0_i}{(\tau^0_i)^2}\frac{\partial \tau^0_i}{\partial\nuta_j}(\p_i-\pi^0_i)-\sum_{i=1}^\n \frac{2\pi^0_i+\tau^0_i}{(\tau^0_i)^3}\delta_i (\frac{\partial \tau^0_i}{\partial\nuta_j})(\sum_k \frac{\partial \tau^0_i}{\partial\nuta_k}(\hnuta_k-\nuta^0_k)+ O_p(m^{-1}))+O_p(m^{-1})\\
&=\sum_{i=1}^\n   \left(\frac{\pi^0_i+\tau^0_i}{(\tau^0_i)^2}\right)\frac{\partial \tau^0_i}{\partial\nuta_j}\delta_i+\sum_{i=1}^\n  \frac{2\pi^0_i}{(\tau^0_i)^2}\frac{\partial \tau^0_i}{\partial\nuta_j}(\p_i-\pi^0_i)-\sum_k (\hnuta_k-\nuta^0_k)\sum_{i=1}^\n \frac{2\pi^0_i+\tau^0_i}{(\tau^0_i)^3}\delta_i \frac{\partial \tau^0_i}{\partial\nuta_j} \frac{\partial \tau^0_i}{\partial\nuta_k}+O_p(m^{-1})
\end{split}\]
The rhs of Eq.~\eqref{eq:lrhs} becomes
\[\begin{split}&rhs=\sum_{i=1}^\n  \frac{\p_i+\htau_i}{\htau_i^2}\frac{\partial \htau_i}{\partial\nuta_j}(\htau_i-\tau^0_i)\\
&=\sum_{i=1}^\n  \frac{\p_i+\htau_i}{\htau_i^2}\frac{\partial \htau_i}{\partial\nuta_j}\left( \sum_k \frac{\partial \tau^0_i}{\partial\nuta_k}(\hnuta_k-\nuta^0_k)+O_p(m^{-1}) \right)\\
&=\sum_k (\hnuta_k-\nuta^0_k)
\sum_{i=1}^\n    \frac{\p_i+\htau_i}{\htau_i^2}\frac{\partial \htau_i}{\partial\nuta_j}\frac{\partial \tau^0_i}{\partial\nuta_k}+O_p(m^{-1})\\
&=\sum_k  (\hnuta_k-\nuta^0_k)
\sum_{i=1}^\n   \left(\frac{\pi^0_i+\tau^0_i}{\tau^0_i}+O_p(m^{-\frac{1}{2}})\right)\left(\frac{1}{\tau^0_i}+O_p(m^{-\frac{1}{2}})\right)\left(\frac{\partial \tau^0_i}{\partial\nuta_j}+O_p(m^{-\frac{1}{2}})\right) \frac{\partial \tau^0_i}{\partial\nuta_k}+O_p(m^{-1})\\
&=\sum_k  (\hnuta_k-\nuta^0_k)
\sum_{i=1}^\n  \frac{\pi^0_i+\tau^0_i}{(\tau^0_i)^2}\frac{\partial \tau^0_i}{\partial\nuta_j} \frac{\partial \tau^0_i}{\partial\nuta_k}
+O_p(m^{-1})\\
\end{split}\]
Hence, equating lhs and  rhs leads to, for $j=1,\cdots, q^*$,
\begin{equation}\label{eq:lr}\begin{split}
&\sum_{i=1}^\n   \left(\frac{\pi^0_i+\tau^0_i}{(\tau^0_i)^2}\right)\frac{\partial \tau^0_i}{\partial\nuta_j}\delta_i+\sum_{i=1}^\n  \frac{2\pi^0_i}{(\tau^0_i)^2}\frac{\partial \tau^0_i}{\partial\nuta_j}(\p_i-\pi^0_i)-\sum_k (\hnuta_k-\nuta^0_k)\sum_{i=1}^\n \frac{2\pi^0_i+\tau^0_i}{(\tau^0_i)^3}\delta_i \frac{\partial \tau^0_i}{\partial\nuta_j} \frac{\partial \tau^0_i}{\partial\nuta_k}+O_p(m^{-1})\\
=&\sum_k  (\hnuta_k-\nuta^0_k)
\sum_{i=1}^\n  \frac{\pi^0_i+\tau^0_i}{(\tau^0_i)^2}\frac{\partial \tau^0_i}{\partial\nuta_j} \frac{\partial \tau^0_i}{\partial\nuta_k}+O_p(m^{-1})\end{split}\end{equation}
Note that, under assumption [A0], all the terms are $O_p(m^{-\frac{1}{2}})$ or smaller except for the first term on the lhs.  Letting $m$ grow to infinity in Eq.~\eqref{eq:lr} implies that \[\sum_{i=1}^\n   \left(\frac{\pi^0_i+\tau^0_i}{(\tau^0_i)^2}\right)(\frac{\partial \tau^0_i}{\partial\nuta_j})\delta_i=0\;\;\text{for all $j$}\,,\]
which in matrix notation becomes 
\begin{equation*}
{B^*}^T \dg\left\{\frac{\pi^0+\tau^0}{{\tau^0}^2} \right\}\delta =0\,,
\end{equation*}
which proves Eq.~\eqref{eq:lemma2} of Lemma~\ref{lemma:htau}.
Plugging this result back into Eq.~\eqref{eq:lr}, we have
\[
0+2{B^*}^T \dg\left\{\frac{\pi^0}{(\tau^0)^2}\right\} \left(\p-\pi^0\right)= {B^*}^T\dg\left\{\frac{(\pi^0+\tau^0)\tau^0+(2\pi^0+\tau^0)(\pi^0-\tau^0)}{(\tau^0)^3}\right\} B^*  (\hnuta-\nuta^0)  +O_p(m^{-1})\,.
\] 
That is
\[2{B^*}^T \dg\left\{\frac{\pi^0}{(\tau^0)^2}\right\} \left(\p-\pi^0\right)= 2{B^*}^T\dg\left\{\frac{(\pi^0)^2}{(\tau^0)^3}\right\} B^*  (\hnuta-\nuta^0)  +O_p(m^{-1})\,.
\]
Hence
\begin{equation}\label{eq:heta}\begin{split} \sqrt{m}(\hnuta-\nuta^0) 
&= ({B^*}^T \dg\left\{\frac{(\pi^0)^2}{(\tau^0)^3}\right\} B^* )^{-1}{B^*}^T\dg\left\{\frac{\pi^0}{(\tau^0)^2}\right\} \sqrt{m}\left(\p-\pi^0\right)  +O_p(m^{-\frac{1}{2}})\\
&=:P^*\sqrt{m}\left(p-\pi^0\right) +O_p(m^{-\frac{1}{2}})\,.
\end{split}\end{equation}
Therefore
\begin{equation}\label{eq:htau}\begin{split} \sqrt{m}(\htau-\tau^0) &= (\frac{\partial \tau^0}{\partial \nuta} +O_p(m^{-\frac{1}{2}}))P^* \sqrt{m}\left(\p-\pi^0\right)+O_p(m^{-\frac{1}{2}})\\
&= B^* P^* \sqrt{m}\left(\p-\pi^0\right)+O_p(m^{-\frac{1}{2}})\\
&=:E^*\sqrt{m}\left(p-\pi^0\right) +O_p(m^{-\frac{1}{2}})\,.
\end{split}\end{equation}
Hence
\[ \sqrt{m}  \begin{pmatrix}
    \p - \pi^0\\[1mm]
    \htau-\tau^0
  \end{pmatrix} 
  = 
   \begin{pmatrix}
    I\\[1mm]
   E^*
  \end{pmatrix} 
   \sqrt{m}\left(\p-\pi^0\right)
  +O_p(m^{-\frac{1}{2}})  \,.\]
  \end{proof}

\end{widetext}

\bibliographystyle{hunsrt}
\bibliography{GausCLs}{}
\end{document}